\documentclass{theoretics}
\usepackage[T1]{fontenc}
\usepackage{amsfonts,amsthm,amsmath,amssymb}
\usepackage{
framed}  

\usepackage[capitalise,noabbrev,nameinlink]{cleveref} 

\usepackage[noend]{algpseudocode} 
\usepackage{xcolor}

\title{Strongly Sublinear Algorithms for Testing Pattern Freeness}

\ThCSauthor{Ilan Newman}{ilan@cs.haifa.ac.il}
\ThCSaffil{Department of Computer Science, University of Haifa, Israel}

\ThCSauthor{Nithin Varma}{nithinvarma@cmi.ac.in}
\ThCSaffil{Max Planck Institute for Informatics, SIC, Saarbrücken, Germany}
\ThCSthanks{Invited paper from ICALP 2022. The first author was supported by the Israel Science Foundation, grant number 379/21. Most of this work was done while the second author was an Assistant Professor at the Chennai Mathematical Institute, India.}

\ThCSyear{2024}
\ThCSarticlenum{1}
\ThCSdoi{10.46298/theoretics.24.1}
\ThCSreceived{Oct 10, 2022}
\ThCSrevised{Oct 26, 2023}
\ThCSaccepted{Nov 20, 2023}
\ThCSpublished{Jan 12, 2024}
\ThCSkeywords{Property testing, Pattern freeness, Sublinear algorithms}

\ThCSshortnames{I.\ Newman and N.\ Varma}
\ThCSshorttitle{Strongly Sublinear Algorithms for Testing Pattern Freeness}

\crefname{observation}{Observation}{Observations}
\crefname{claim}{Claim}{Claims}

\newcommand{\nvedit}[1]{{\color{magenta} #1}}

\newcommand{\den}{\mathsf{den}}
\DeclareMathOperator{\polylog}{polylog}

\DeclareMathOperator{\Hdist}{Hdist}
\DeclareMathOperator{\Ddist}{Ddist}
\DeclareMathOperator{\dist}{dist}
\DeclareMathOperator{\St}{St}
\DeclareMathOperator{\NE}{NE}
\newcommand{\eps}{\varepsilon}
\newcommand{\seq}{\mathsf{seq}}
\newcommand{\R}{\mathbb{R}}
\newcommand{\N}{\mathbb{N}}
\newcommand{\bx}{\mathsf{box}}
\newcommand{\cC}{\mathcal{C}}

\DeclareMathOperator{\E}{\textbf{E}}

\newcommand{\cS}{\mathcal{S}}
\newcommand{\ignore}[1]{}

\begin{document}

\maketitle
\begin{abstract}
For a permutation $\pi:[k] \to [k]$, a function $f:[n] \to \R$
contains a $\pi$-appearance if there exists $1 \leq i_1 < i_2 < \dots < i_k
\leq n$  such that for all $s,t \in [k]$, $f(i_s) < f(i_t)$ if and only if $\pi(s) < \pi(t)$. The function is $\pi$-free if it has no $\pi$-appearances. 
In this paper, we investigate the problem of testing whether an input function $f$ is $\pi$-free or whether $f$ differs on at least~$\eps n$ values from every $\pi$-free function. 
This is a generalization of the well-studied monotonicity testing and was first studied by Newman, Rabinovich, Rajendraprasad and Sohler~\cite{NewmanRRS19}. We show that for all constants $k \in \mathbb{N}$, $\eps \in (0,1)$, and permutation $\pi:[k] \to [k]$, there is a one-sided error $\eps$-testing algorithm for $\pi$-freeness of functions $f:[n] \to \R$ that makes $\tilde{O}(n^{o(1)})$ queries.
We improve significantly upon the previous best upper bound $O(n^{1 - 1/(k-1)})$ by Ben-Eliezer and Canonne~\cite{Ben-EliezerC18}. 
Our algorithm is adaptive, while the earlier best upper bound is known to be tight for nonadaptive algorithms. 
\end{abstract}



\section{Introduction}\label{sec:intro}
Given a permutation $\pi:[k] \to [k]$, a
function $f:[n] \to \R$ contains a $\pi$-appearance if there exists
$1 \leq i_1 < i_2 < \dots < i_k \leq n$  such that for all $s,t \in [k]$ it holds that $f(i_s) < f(i_t)$ if and only if $\pi(s) < \pi(t)$. 
In other words, the function values restricted to the indices $\{i_1, \dots,i_k\}$ respect the ordering in $\pi$. 
The function is $\pi$-free if it has no $\pi$-appearance. 
For instance, the set of all real-valued monotone non-decreasing functions over $[n]$ is $(2,1)$-free.
The notion of $\pi$-freeness is well-studied in combinatorics, where the famous Stanley-Wilf conjecture about the bound on the number of $\pi$-free permutations $f:[n] \to [n]$ has spawned a lot of work~\cite{Bona1997,Bona99,Arratia99,Klazar00,AlonF00}, ultimately culminating in a proof by Marcus and Tardos~\cite{MarcusT04}. 
The problem of designing algorithms to determine whether a given
permutation $f:[n] \to [n]$ is $\pi$-free is an active area of
research~\cite{AlbertAAH01,AhalR08,BerendsohnKM2021}, with 
linear time algorithms  for constant $k$  \cite{GuillemotM14,Fox13}.
Apart from the theoretical interest, practical motivations to study $\pi$-freeness include the study of motifs and patterns in time series analysis~\cite{DBLP:conf/kdd/BerndtC94,DBLP:conf/icdm/PatelKLL02,DBLP:conf/kdd/KeoghLC02}.

In this paper, we study property testing
of $\pi$-freeness,  as proposed by Newman, Rabinovich,
Rajendraprasad and Sohler~\cite{NewmanRRS19}. Specifically, given
$\eps \in (0,1)$, an $\eps$-testing algorithm for $\pi$-freeness
accepts an input function $f$ that is $\pi$-free, and rejects if $f$ differs from every $\pi$-free function on at
least $\eps n$ values.\footnote{Algorithms in this area are typically randomized, and the decisions to
  accept or reject are with high constant probability. See~\cite{RubinfeldS96,GoldreichGR98} for definitions of property testing.} The algorithm is given oracle access to the function $f$ and the goal is to minimize the number of queries made by the algorithm.
This problem is a generalization of the well-studied monotonicity
testing on the line ($(2,1)$-freeness), which was one of the
first works in combinatorial property testing, and is still being 
studied actively~\cite{DGLRRS99,EKK+00,BGJRW12,CS13a,Belovs18,DixitRTV18,PallavoorRV18}.

Newman, Rabinovich, Rajendraprasad and Sohler~\cite{NewmanRRS19} showed that for a general permutation $\pi$ of length $k$, the property of $\pi$-freeness can be
$\eps$-tested using a nonadaptive\footnote{An algorithm whose queries
  do not depend on the answers to previous queries is a nonadaptive
  algorithm. It is adaptive otherwise.} algorithm of query complexity
$O_{k,\eps}(n^{1 - 1/k})$.\footnote{Throughout this work, we are
  interested in the parameter regime of constant $\eps \in (0,1)$ and
  $k$. The notation $O_{k,\eps}(\cdot)$ hides a factor that is an
  arbitrary function of these parameters.} Additionally, they showed
that, for nonadaptive algorithms, one cannot obtain a significant
improvement on this upper bound for $k \geq 4$. 
In a subsequent work, Ben-Eliezer and Canonne~\cite{Ben-EliezerC18} improved this upper bound to $O_{k,\eps}(n^{1 - 1/(k-1)})$, which they showed to be tight for nonadaptive algorithms.
For monotone permutations $\pi$ of length $k$, namely, either $(1,2,\dots,k)$ or $(k,k-1,\dots,1)$, Newman et al.~\cite{NewmanRRS19} presented an algorithm with query complexity $(\eps^{-1}\log n)^{O(k^2)}$ to $\eps$-test $\pi$-freeness. The complexity was improved, in a sequence of works~\cite{Ben-EliezerCLW19,Ben-EliezerLW19}, to $O_{k,\eps}(\log n)$, which is optimal for constant $\eps$ even for the special case of testing $(2,1)$-freeness~\cite{Fischer04}. 

Despite the aforementioned advances in testing freeness of
monotone permutations, improving the complexity of testing freeness of
arbitrary permutations has remained open all this while.
For arbitrary permutations of length at
most $3$, Newman et al.~\cite{NewmanRRS19} gave an adaptive algorithm
for testing freeness 
with query complexity $(\eps^{-1} \log n)^{O(1)}$.  However, the case
of general $k > 3$ has remained elusive. In particular, the
techniques of \cite{NewmanRRS19} for $k=3$ do not seem to generalize
even for $k=4$.

As remarked above, optimal \emph{nonadaptive} algorithms are known
 for any $k$~\cite{Ben-EliezerC18}, but,  their complexity tends to
 be linear in the input length as
$k$ grows. 
For the special case of $(2,1)$-freeness, it is well-known that adaptivity does not help at all in improving the complexity of testing~\cite{EKK+00,Fischer04}. 
Adaptivity is known to help somewhat for the case of testing freeness of monotone
permutations of length $k$, where, every nonadaptive algorithm has query
complexity $\Omega((\log n)^{\log k})$~\cite{Ben-EliezerCLW19}, and
the $O_{k,\eps}(\log n)$-query algorithm of Ben-Eliezer, Letzter, and Waingarten~\cite{Ben-EliezerLW19} is adaptive. Adaptivity significantly helps
in testing freeness of arbitrary permutations of length~$3$
as shown by \cite{NewmanRRS19} and \cite{Ben-EliezerC18}.

\subsection{Our results} In this work, we give adaptive $\eps$-testing
algorithms for $\pi$-freeness of permutations $\pi$ of arbitrary
constant length $k$ with complexity
 $\tilde{O}_{k,\eps}(n^{o(1)})$. 
 Hence,  testing $\pi$-freeness has quite efficient 
sublinear algorithms even for relatively large patterns. 
Our result shows a strong separation between adaptive and nonadaptive
algorithms for testing pattern freeness.

\begin{theorem}\label{thm:main}
Let $\eps \in (0,1),$ $k \in \mathbb{N}$ and $\pi:[k] \to [k]$ be a permutation. There exists an $\eps$-tester for $\pi$-freeness of functions $f:[n] \to \mathbb{R}$ with query complexity $\tilde{O}\left(\left(\frac{k}{\eps}\right)^{\Theta(\log \log n)}n^{k/\log \log \log n}\right)$. 
\end{theorem}

\subsection{Discussion of our techniques}
Our algorithm has one-sided error and rejects only if it finds a $\pi$-appearance in the input function $f:[n] \to \R$.
In the following, we present some of the main  ideas behind a $\tilde{O}(\sqrt{n})$-query algorithm
for detecting a
$\pi$-appearance in a function $f$ that is $\eps$-far from $\pi$-free,
for a permutation $\pi$ of length $4$. The case of length-$4$ permutations is not very
different from the general case (where, we additionally recurse
on problems of smaller length patterns). The $\tilde{O}(\sqrt{n})$
queries algorithm is much simpler than the general one, but
it outlines many of the ideas involved in the latter. Additionally, it already beats the lower bound of 
$\Omega(n^{2/3})$ on the complexity of nonadaptive algorithms for $\pi$-freeness testing patterns of length $4$~\cite{Ben-EliezerC18}. A more detailed
description appears in~\cref{sec:top-level}. The formal description of the general
algorithm is given in~\cref{sec:generalizedtesting}. 
  
For a parameter $\eps \in (0,1)$, a function $f$ is $\eps$-far
from $\pi$-free if at least $\eps n$ values of $f$ need to be changed in order to make it $\pi$-free.
In other words, the Hamming distance of $f$ to the closest real-valued $\pi$-free  function over $[n]$ is at least $\eps n$. 
A folklore fact is that the Hamming distance and
the deletion distance of $f$ to $\pi$-freeness
are equal, where the deletion distance of $f$
to $\pi$-freeness is the cardinality of the smallest set $S \subseteq [n]$ such that
$f$ restricted to $[n]\setminus S$ is $\pi$-free. 
By virtue of this equality, a function that is $\eps$-far from $\pi$-free has a matching of $\pi$-appearances of cardinality at least $\eps n/4$, where a matching of $\pi$-appearances is a collection of $\pi$-appearances such that no two of them share an index.
This observation 
facilitates our algorithm and all previous algorithms on testing
$\pi$-freeness, including monotonicity testers. 

The basic ingredient in our algorithms is the use of a natural representation of  $f:[n] \to \R$ by a Boolean function over a grid $[n] \times R(f)$, where $R(f)$ denotes the range of $f$.  
Specifically, we visualize the function as a grid of $n$ points in
$\R^2$, such that for each $i \in [n]$, the pair $(i,f(i))$ is a
point 
of the grid. We use $G_n$ to denote this grid of points. This view has been useful in the design of approximation
algorithms for the related and fundamental problem of estimating the
length of Longest Increasing Subsequence (LIS) in a real-valued array  \cite{SaksS10,RubinsteinSSS19,MitzenmacherS21,NV20}.
Adopting this view, for any permutation $\pi:[k] \to [k]$, a $\pi$-appearance at
$(i_1, \ldots ,i_k)$ in $f$ corresponds naturally to a
$k$-tuple of points $(i_j, f(i_j)), ~ j=1
\ldots k$ in $G_n,$ for which their relative order (in $G_n$) forms a  
$\pi$-appearance.   The converse
is also true: every $\pi$-appearance in the Boolean grid
$G_n$ corresponds to a $\pi$-appearance in $f$.

We note that the grid $G_n$ is neither known to, nor directly accessible by, the
algorithm. In particular, $R(f)$ is not assumed to be known. A main
first step in our algorithm is to approximate the grid $G_n$ by a
 coarser $m \times m$ grid $G_{m,m}$ of boxes, for  a parameter $m = o(n)$ that will determine the query complexity. 
   The grid $G_{m,m}$ is defined as follows. Suppose that we have a
   partition of $R(f)$ into $m$ disjoint contiguous intervals of
   increasing values, referred to here as `layers', $I_1,
   \ldots ,I_m$, and let $S_1, \ldots ,S_m$ be a partition of $[n]$
   into $m$ contiguous intervals of equal size, referred to as `stripes'. These two partitions
   decompose  $G_n$ and the $f$-points in it into $m^2$
   boxes and forms the grid $G_{m,m}$. The
   $(i,j)$-th cell of this grid is the Cartesian product $S_i \times
   I_j$, and is denoted $\bx(S_i, I_j)$.
We view the nonempty boxes in $G_{m,m}$ as a coarse
approximation of $G_n$ (and of the input function, equivalently). The
grid $G_{m,m}$ has a natural order on its boxes (viewed as points in
$[m] \times [m]$).
   
While $G_{m,m}$ is also not directly accessible to the algorithm, it
can be well-approximated very efficiently. 
We can do this by sampling $\tilde{O}(m)$ indices from $[n]$ independently and uniformly at random and making queries to those indices to identify and \emph{mark} the boxes in $G_{m,m}$ that contain a non-negligible
density of points of $G_n$. This provides a good enough
approximation of the grid $G_{m,m}$. 
For the rest of this high-level explanation, assume that we have fixed
$m << n$, and we know $G_{m,m}$; that is, we assume that we know the number of
points of $G_n$ belonging to each  box in $G_{m,m}$, but not necessarily the points themselves.

If we find $k$ nonempty boxes in $G_{m,m}$ that form a
$\pi$-appearance when viewed as points in the $[m] \times [m]$ grid, then $G_n$ (and
hence $f$) contains a $\pi$-appearance for any set of $k$ points that
is formed by selecting one point from each
of the corresponding boxes. 
See \cref{fig:1}(A) for
such a situation, for $\pi = (3,2,1,4)$.  We first detect such
$\pi$-appearances by our knowledge of $G_{m,m}$. However, the converse
is not true: it could be that $G_n$ contains many $\pi$-appearances,
where the corresponding points, called `legs', are in
boxes that share layers or stripes, and hence do not form
$\pi$-appearances in $G_{m,m}$.
  See e.g.,
 \cref{fig:1}(B) for such an appearance for $\pi = (3,2,1,4)$.
Thus, if the function is far from being $\pi$-free and no
$\pi$-appearances are detected in $G_{m,m}$, then there must be many
$\pi$-appearances in which some legs share a layer or a
stripe in $G_{m,m}$. In this case,  the seminal result of Marcus and Tardos~\cite{MarcusT04},
implies that only $O(m)$ of the boxes in $G_{m,m}$ are nonempty. 
An averaging argument implies that if $f$ is $\eps$-far from
being $\pi$-free, then after deleting layers or
stripes in $G_{m,m}$  with $\omega(1)$-dense boxes, we are still left with a partial function
(on the undeleted points) that is $\eps'$-far from being
$\pi$-free, for a large enough $\eps'$.
 
For the following high-level description we consider $\pi=(3,2,1,4)$, although all the
following ideas work for any permutation of length $4$. 
Any $\pi$-appearance has its four legs spread over at most $4$ marked boxes. 
This implies that there are only constantly many non-isomorphic ways
of arranging the marked boxes containing any particular
$\pi$-appearance, in terms of the order relation among the marked boxes,
and the way the legs of the $\pi$-appearance are included in them. 
These constantly many ways are called `configurations' in the
sequel. Thus any $\pi$-appearance is consistent with a certain
configuration. Additionally, in the case that  
multiple points in a $\pi$-appearance share some marked boxes,
this appearance induces the appearances of permutations of length smaller than $4$ in each box (which are sub-permutations $\nu$ of $\pi$). 
If a constant fraction of the $\pi$-appearances are spread across multiple marked boxes, there will be many such $\nu$-appearances in the marked boxes in the coarse grid.
Hence, one phase of our algorithm will run tests for $\nu$-appearances
for smaller patterns $\nu$ (which can be done in $\polylog n$ queries
using known testers for patterns of length at most $3$) on each marked
box, and combine these $\nu$-appearances to detect a $\pi$-appearance,
if any.  This phase, while seemingly simple will require extra care,
as combining sub-patterns appearances into a global $\pi$-appearance is
not always possible. This is a major issue in the general
case for $k > 4$. 

The simpler case is when there is a constant fraction of
$\pi$-appearances such that all $4$ points of each such appearance
belong to a single marked box. This can be solved by randomly sampling
a few marked boxes and querying all the points in them to see if there
are any $\pi$-appearances. The
case that a constant fraction of the $\pi$-appearances have their legs belonging to the same
layer or the same stripe is an easy extension of
the `one-box' case.

To obtain the desired query complexity, consider first setting $m =
\tilde{O}(\sqrt{n})$. Getting a good enough estimate of $G_{m,m}$ as
described above take $\tilde{O}(m)=\tilde{O}(\sqrt{n})$ queries. Then, testing each
box for $\nu$-freeness, for smaller permutations $\nu$ takes
$\polylog n$ per test, but since this is done for all marked
boxes, this step also takes $\tilde{O}(m)=\tilde{O}(\sqrt{n})$. Finally, in the last
simpler case, we may just query all indices in a sampled box that contains at most
$n/m = \Theta(\sqrt{n})$ indices, by our setting of~$m$. This results in a
$\tilde{O}(\sqrt{n})$-query tester for $\pi$-freeness.

To obtain a better complexity, we reduce the value of $m$, and, in the last step, we randomly sample 
a few marked boxes and run the algorithm recursively. 
This is so, since, in the last step, we are in the case that for a
constant fraction of the $\pi$-appearances, all four legs of each
$\pi$-appearance  belong to a single marked box (or a constant number
of marked boxes sharing a layer or stripe). The depth of recursion
depends monotonically on  $n/m$ and the larger it is the
smaller is the query complexity. 
The bound we describe in this article
is $n^{O(1/\log \log \log n)}$ which is due to the exponential deterioration of the
distance parameter $\eps$ in each recursive call.
Our algorithm for permutations of length $k > 4$  uses, in
addition to the self-recursion, a 
recursion on $k$ too.

Finally, 
we call $\nu$-freeness or $\pi$-freeness algorithms on
marked boxes (or a collection of constantly many marked boxes sharing a layer or stripe) and not the entire grid. Since we do not know which points belong to the marked boxes, but 
only know that their density is significant, we can access points in them only via sampling and treating 
points that fall outside the desired box as being \emph{erased}.
This necessitates the use of erasure-resilient
testers~\cite{DixitRTV18}. Such testers are known  for all permutation patterns of length at most $3$~\cite{DixitRTV18,NV20,NewmanRRS19}. 
In addition, the basic tester we design is also erasure-resilient,
which allows us to recursively call the tester on appropriate subsets of marked boxes.


\paragraph{Some additional challenges we had to overcome:} In the recursive
algorithm for $k$-length permutation freeness, $k \geq 4$, we need to find $\nu$-appearances  that are restricted to appear in specific
configurations, for smaller
length permutations $\nu$. To exemplify this notion, consider testing
$\nu=(1,3,2)$-freeness. In the unrestricted setting, $f: [n]
\mapsto \R$ has a $\nu$-appearance if the values at any three indices have a
$\nu$-consistent order. 
In a restricted setting, we may ask ourselves
whether $f$ is free of $\nu$-appearances where the indices
corresponding to the $1,3$-legs of a $\nu$-appearance are of value at most $n/2$ (that is in
the first half of $[n]$), while the index corresponding to the $2$-leg is larger than
$n/2$. This latter property seems at least as hard to test as the
unrestricted one.  In particular, for the $\nu$-appearance as described above, it could be that while $f$ is far
from being $\nu$-free in the usual sense, it is still free of having restricted
$\nu$-appearances.  In our algorithm, we need to test (at lower
recursion levels) freeness from such restricted appearances. This extra
restriction is discussed at a high level in~\cref{sec:top-level}. For a formal definition of the 
restricted testing problem and how it fits into our final algorithm, see~\cref{sec:generalizedtesting}.

\subsection{Open directions}
\paragraph{Testing restricted $\pi$-freeness:}
Testing for restricted $\pi$-appearance, as described above,  is 
at least as hard as testing $\pi$-freeness. For monotone patterns (and hence
$2$-patterns) testing freeness and testing restricted appearances are
relatively easy (can be done in $\polylog n$
queries). 
  For patterns of size $3$ and more, the complexity
of testing freeness of restricted appearances is currently open.

\paragraph{Weak $\pi$-freeness:}
In the definition of
$\pi$-freeness, we required strict inequalities on function values to
have an occurrence of the pattern. A natural variant is to allow weak
inequalities, that is -- 
for a set indices $1 \leq i_1 < i_2 \dots < i_k \leq n$  a {\em weak}-$\pi$-appearance is
when 
for all $s,t \in [k]$ it holds that $f(i_s) \leq f(i_t)$ if and only if $\pi(s) < \pi(t)$. 
Such a relaxed requirement would mean that having
a collection of $k$ or more equal values  is already a $\pi$-appearance
for any pattern $\pi$.  For monotone patterns of length $k$, the deletion
distance equals the Hamming distance, for any $k$, for this relaxed
definition as well. We do not know if this is true for larger $k$ for
non-monotone patterns in general, although we suspect
that the Hamming distance is never larger than the deletion distance
by more than a constant factor. Proving this will be enough to make our
results true for testing freeness of any constant size forbidden permutation, even with the relaxed definition. We show that the Hamming distance is equal to the deletion
distance for patterns of length at most $4$. 
 Hence,~\cref{thm:main} also holds for weak-$\pi$-freeness for $k \leq
 4$.

Another similarly related variant  is when the forbidden order pattern is
not necessarily a permutation (that is, an arbitrary function from $[k]$ to $[k]$ which is not one-to-one). 
For example, for the $4$-pattern
$\alpha=(1,2,3,1)$, an $\alpha$-appearance in $f$ at indices $i_1 <
i_2 < i_3 < i_4$ is when $f(i_1) < f(i_2) < f(i_3)$ and $f(i_4) =
f(i_1)$, as dictated by the order in $\alpha$.  For testing freeness of such patterns, an $\Omega(\sqrt{n})$ adaptive lower bounds exist (by a
simple probabilistic argument) even for the very simple case of
$(1,1)$-freeness, which corresponds to the property of being a one-to-one function. 

An interesting point to mention, in this context, is that for testing freeness of forbidden
permutations, a major tool that we use is the Marcus-Tardos
bound~\cite{MarcusT04}. Namely,  
that the number of $1$'s in an $m \times m$ Boolean matrix that does not contain a specific
permutation matrix of order $k$ is $O(m)$. For
non-permutation patterns, similar bounds are not true in general
anymore, but  do hold in many cases (or hold in a weak sense, e.g.,
only slightly more than linear).
 In such cases, the Marcus-Tardos bound
could have allowed relatively efficient
testing. However,  the lower bounds hinted above for the $(1,1)$-pattern makes the
testing problem 
completely different from that of testing forbidden permutation
patterns.

\paragraph{Restricted functions:}
In this paper we always consider the set of functions $f:[n] \mapsto
\R$ with no restrictions.  Interesting questions occur when the set of
functions is more restricted.  One natural such restriction is for functions of bounded or
restricted range (for the special case of $(2,1)$-freeness, such a study was initiated by Pallavoor, Raskhodnikova and Varma~\cite{PallavoorRV18} and followed upon by others~\cite{Belovs18,NV20}). 
We do know that
in the very extreme case, that is, for functions from the line $[n]$ to a
constant-sized range,  pattern freeness is testable in constant time
even for much more general class of forbidden
patterns~\cite{AlonKNS00}. Apart from this extreme restriction,
or the results for $2$-patterns stated above, 
we are not aware of results concerning functions of bounded range
(e.g., range that is $n^2$ or $\sqrt{n}$). 

Lastly, if we restrict our attention to functions $f:[n] \to [n]$ that are themselves
permutations, Fox and Wei~\cite{FoxW18} argued that for some special types of distance measures
such as the rectangular-distance and Kendall's tau distance,
testing $\pi$-freeness can be done in constant query complexity.
Testing $\pi$-freeness w.r.t.\ the Hamming or deletion distances is
very different, and 
still remains open for this setting.

\paragraph{Other open questions:} 
The major open question left in this paper 
is to determine the exact (asymptotic)  complexity of testing
$\pi$-freeness of arbitrary permutations $\pi:[k] \to [k], ~ k \geq
3$. While the gaps for $k=3$ are relatively small (within $\polylog n$
range), the gaps are yet much larger for $k \geq 4$. 
We do not have any reason to think that the upper bound obtained
in this draft is tight. We did not try to optimize the exponent of $n$
in the $\tilde{O}(n^{o(1)})$ expression, but the current methods
do not seem to bring down the query complexity to $\polylog n$. We conjecture, however,
that the query complexity is  $\polylog n$
for all constant $k$.
Another open question is whether the complexity of two-sided
error testing might be lower than that of one-sided error testing.

Finally, Newman and Varma~\cite{NV20} used  lower bounds on
testing pattern freeness of monotone patterns of length $k\geq 3$ (for
nonadaptive algorithms), to obtain lower bounds on the query complexity of 
nonadaptive algorithms for LIS estimation. Proving any lower bound better than $\Omega(\log n)$ for
adaptively testing freeness, for arbitrary permutations of length $k$ for $k \geq
3,$ may translate in a similar way to lower bounds on adaptive algorithms for
LIS estimation.

\paragraph{Organization:}
~\cref{sec:prelim} contains the notation, important definitions, and a discussion of some
key concepts related to testing $\pi$-freeness. 
\cref{sec:top-level} contains a high-level overview of an $\tilde{O}(\sqrt{n})$-query algorithm for patterns of length $4$.
The formal description of our $\pi$-freeness tester for permutations $\pi$ of length $k \geq 4$ and the proof of correctness appear in~\cref{sec:generalizedtesting}.

\section{Preliminaries and discussion}\label{sec:prelim}


For a function $f:[n] \to \R$, we denote by $R(f)$ the image of $f$.
We often refer to the elements of the domain $[n]$ as \emph{indices},
and the elements of $R(f)$ as \emph{values}.
For $S \subseteq [n]$, 
 $f|_{S}$ denotes the restriction of $f$ to $S$.
Throughout, $n$ will denote the  domain size of the function $f$.

We often refer to events in a probability space. For ease of
representation, we will say that
an event $E$ occurs with high probability, denoted `w.h.p.', if $\Pr(E) >
1 - n^{-\log n}$, to avoid specifying accurate constants.

Let $\mathcal S_k$ denote the set of all permutations of length~$k$. 
We
view $\pi = (a_1, \ldots ,a_k) \in \mathcal{S}_k$ as a function (and
not as a cyclus), that is, where $\pi(i)  = a_i, ~  i \in [k]$.
We refer to $a_i$ as the $i$th value in~$\pi$, and as the
$a_i$-leg of~$\pi$. Thus, e.g., for $\pi=(4,1,2,3),$ the first value is
$4$, and the third is $2$, while the $4$-leg of~$\pi$ is at the first
place and its $1$-leg is at the second place.
We often refer to $\pi \in \mathcal S_k$ as a $k$-pattern.

\subsection{Deletion distance  vs.\ Hamming distance}

 The distance of a function from the property of being $\pi$-free can
 be measured in several  ways. In this paper, we use 
  Hamming and deletion distances as  
  are defined next.

\begin{definition}[Deletion and Hamming distance]
Let $f: [n] \to \R$. The deletion distance of $f$ from being $\pi$-free
is $\Ddist_{\pi}(f) = \min \{|S|: ~ S \subseteq [n], ~
f|_{[n]\setminus S} ~ \mbox
{is } \pi\mbox{-free} \}$. Namely, it is the cardinality of the smallest set $S \subseteq
[n]$ that intersects each $\pi$-appearance in $f$.
The Hamming distance of $f$ from being $\pi$-free, $\Hdist_{\pi}(f)$ is the minimum of $\dist(f,f') = |\{i:~ i \in [n],~ f(i) \neq f(i') \}|$ over all
functions $f': [n] \to \R$ that are $\pi$-free.
\end{definition}
For $0 \leq \eps <1$ we say that $f$ is $\eps$-far from $\pi$-freeness in deletion
distance, or Hamming distance, if $dist_{\pi}(f)  \geq
\eps n$, and otherwise we say that $f$ is $\eps$-close to
$\pi$-freeness, where $dist_{\pi}(f)$ is the corresponding distance.

\begin{claim}\label{cl:ham-del}
$\Ddist_{\pi}(f) = \Hdist_{\pi}(f)$
\end{claim}

\begin{proof}
It is obvious from the definition that $\Ddist_{\pi}(f) \leq
  \Hdist_{\pi}(f)$. For the other direction, assume that
  $\Ddist_{\pi}(f) = d$. Let $S = \{i_1, i_2, \dots, i_d\} \subseteq [n]$ for $i_1 < i_2 \ldots < i_d$ be such that $f|_{[n]
    \setminus S}$ is $\pi$-free. 
    If
  $i_1 >1$, 
  consider the function $f':[n] \to \R$ such that for $i \notin S$, $f'(i) = f(i)$ and for $j \in [d]$, $f'(i_j) = f(m_j)$, where $m_j$
  is the largest element in $[i_j] \setminus S$. It can be seen that $f'$ is $\pi$-free.
  Moreover, $\Hdist(f,f') \leq d$, which proves the
  claim for $S$ such that $i_1 > 1$. 
    If $i_1 = 1$, let $s \in [n]$ be the smallest index not in $S$.
    We consider the function $f'':[n] \to \R$, where $f''(i) = f(s)$ for all $i \in [s-1]$, where $[s-1] \subseteq S$, by definition of $s$.
    Now, the deletion distance of $f''$ is less than $d$ and we are back to the case that the smallest index being deleted is greater than $1$.
\end{proof}

\cref{cl:ham-del} is extremely important for testing
$\pi$-freeness, and is what gives rise to \emph{all} testers of
monotonicity, as well as $\pi$-freeness that are known. This is due to the fact
that the tests are really designed for the deletion distance, rather than
the Hamming distance. The folklore observation made
in~\cref{clm:matching} facilitates such tests, and~\cref{cl:ham-del} makes the tests work also for the Hamming distance.
Due to~\cref{cl:ham-del}, we say that a function $f$ is
$\eps$-far from $\pi$-free without specifying the 
distance measure.

  Let $\pi \in {\mathcal S}_k$  and $f:[n] \to \R$. A {\em matching} of
  $\pi$-appearances in $f$ is a collection of $\pi$-appearances 
  that are pairwise disjoint as sets of indices in $[n]$. 
The following claim is folklore and immediate from the fact that the size of a minimum
vertex cover of a $k$-uniform hypergraph is at most $k$ times the 
cardinality of a maximal matching.

\begin{claim}\label{clm:matching}
Let $\pi \in \mathcal{S}_k$. If  $f:[n] \to \R$  is $\eps$-far from being
$\pi$-free, then there exists a matching of $\pi$-appearances of size at least $\eps n /k$. 
\end{claim}

All our algorithms have one-sided error, i.e., they always accept
functions that are $\pi$-free.  For functions that are far from being
$\pi$-free, using~\cref{clm:matching}, our algorithms aim to
detect some $\pi$-appearance, providing a witness for the function to
not be $\pi$-free. Hence, in the description below, and throughout the
analysis of the algorithms, the input function is assumed to be
$\eps$-far from $\pi$-free.

\subsection{Viewing a function as a grid of points}

Let $f: [n] \to \R$. We view $f$  as points in an
$n \times |R(f)|$ grid $G_n$. 
The horizontal axis of $G_n$ is labeled with the
indices in $[n]$. 
The vertical axis of $G_n$ represents the image $R(f)$ and is labeled with the
distinct values in $R(f)$ in  increasing order, $r_1 < r_2 < \ldots <r_{n'}$, where
$|R(f)| = n' \leq n$. 
We refer to an index-value pair $(i,f(i)), i \in [n]$ in the grid as a \emph{point}.
The grid has $n$ points, to which our algorithms do not have direct
access. In particular, we do not assume that $R(f)$ is known. 
The function is one-to-one if $|R(f)|=n$. 
Note that if $M$ is a matching of $\pi$-appearances in $f$, then $M$ defines a
  corresponding matching of $\pi$-appearances in $G_n$. We will
 always consider this alternative view, where the matching $M$ is a set of disjoint
 $\pi$-appearances in the grid $G_n$.

\subsubsection{Coarse grid of boxes} 
For a pair of subsets $(S,I)$, where $S \subseteq [n]$ and
$I \subseteq R(f)$, we denote by $\bx(S,I)$, the subgrid $S \times I$
of $G_n$ alongwith with the set $\{(i,f(i)): i \in S, f(i) \in I\}$ of points  in $G_n$.
In most cases, $S$ and~$I$ will be intervals in $[n]$
and $R(f)$, respectively, and hence the name {\em box}.  The
\emph{size} of $\bx(S,I)$ is defined to be $|S|$.  A box is
\emph{nonempty} if it contains at least one point and is \emph{empty}
otherwise.

Consider an arbitrary collection of pairwise disjoint contiguous value intervals
${\mathcal L} = \{I_1, \ldots I_{m}\}$, such that $I \subseteq \bigcup_{j \in [m]} I_j$. The set ${\mathcal L}$ naturally defines a partition of the points in $\bx(S,I)$ into $m$ horizontal \emph{layers}, $\bx(S, I_j)$ for $j\in[m]$.
Assume that, in addition to a set of layers 
 $\mathcal L,$ we have a partition of $S$ into disjoint
  intervals $S = \bigcup_{i=1}^{m}  S_i$ where $S_i= [a_i,b_i],$ and
  $b_{i} < a_{i+1},~ i=1, \ldots m-1$. The family ${\mathcal S} = \{S_1,
  \ldots S_m\}$ induces a partition of $\bx(S,I)$ and the
  points in it, into $m$ vertical \emph{stripes}, $\bx(S_i, I)$ for $i \in [m]$. 
  The layering defined by
  $\mathcal L$  together with the stripes defined by $\mathcal S$ partition
  $\bx(S,I)$ into a coarse grid $G_{m,m}$ of boxes $\{
  \bx(S_i,I_j)\}_{i,j \in [m]}$ that is isomorphic
  to the grid $[m] \times [m]$. Note that $\bx(S,I)$ could even be the entire grid $G_n$.   
  Given such a gridding of $\bx(S,I)$, the layer of $\bx(S_i,I_j)$, denoted $L(\bx(S_i,I_j))$, is $\bx(S,I_j)$ and
  its stripe, denoted  $\St(\bx(S_i,I_j))$, is $\bx(S_i,I)$.

We say that \emph{layer $L$ is below layer
  $L'$}, and write $L < L',$  if the largest value of a point in $L$
is less than the smallest value of a point in $L'$.
For stripes $\St(S), \St(S')$, we write $\St(S) < \St(S')$
if the largest index in $S$ is smaller than the smallest index in $S'$. 
For the grid $G_{m,m}$ and two boxes $B_1, B_2$ in it, $B_1 < B_2$
if $L(B_1) < L(B_2)$ and $\St(B_1) < \St(B_2)$.



\subsubsection{Patterns among and within nonempty boxes}

Consider a coarse grid of boxes, $G_{m,m}$, defined as above on the grid of points $G_n$.
 There is a natural homomorphism from the points in $G_n$ to the
 nonempty boxes in $G_{m,m}$ where those points fall.
For $f$ and a grid of boxes
 $G_{m,m}$ as above, we refer to this homomorphism
 implicitly. 
This  homomorphism defines when 
$G_{m,m}$ contains a $\pi$-appearance in a natural way.  For example,
consider the permutation $\pi = (3,2,1,4)\in \mathcal
S_4$. We say that \emph{$G_{m,m}$ contains $\pi$} if there are nonempty boxes
$B_1, B_2, B_3, B_4$ such that $\St(B_1) < \St(B_2) < \St(B_3) < \St(B_4)$ and
$L(B_3) < L(B_2) < L(B_1) < L(B_4)$ (see~\cref{fig:1}(A)).

\begin{observation}\label{obs:homo}
   Let $\mathcal L, \mathcal S$ be a partition of $G_n$ into layers and stripes
   as above, with $|\mathcal L | =m, ~ |\mathcal S |=m$. If
   $G_{m,m}$ contains $\pi$ then $G_n$ (and equivalently $f$) has a $\pi$-appearance. 
 \end{observation}

 The converse of~\cref{obs:homo} is not true;  $G_n$ may contain a $\pi$-appearance while $G_{m,m}$ does not. This happens when some of the boxes that contain the $\pi$-appearance  share a layer or a stripe. Two boxes are \emph{directly-connected} if they share a layer or a
stripe.
The transitive closure of the relation \emph{directly-connected} is called
\emph{connected}. An arrangement of boxes where every two boxes are connected is called a \emph{connected component}, or simply, a component. The size of a connected component is the number of boxes in it.


\begin{figure}
\begin{center}
\includegraphics[scale = 0.7]{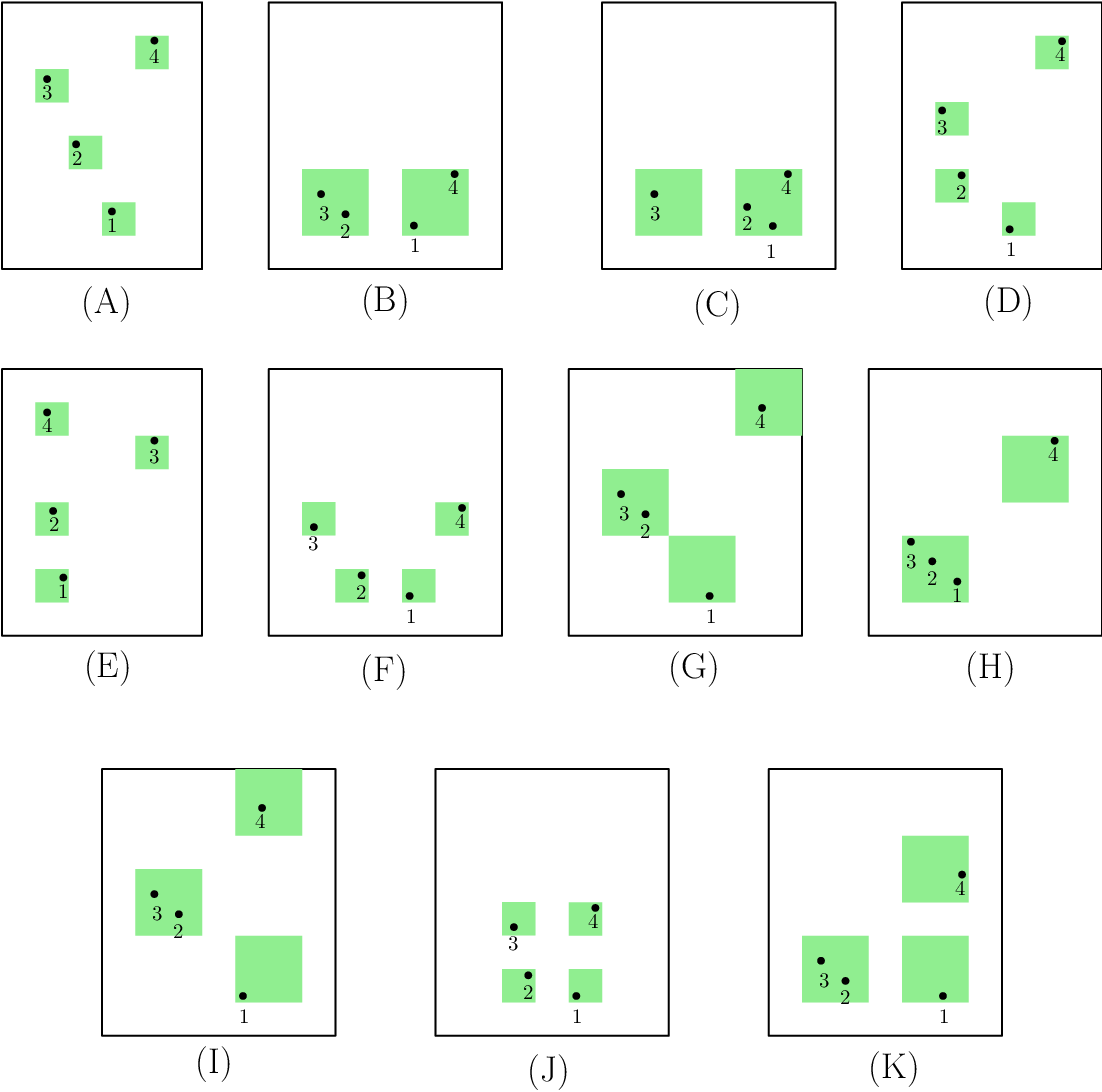}
\end{center}
\caption{Each rectangle  represents a different grid $G_n$, where the
  green shaded boxes correspond to some nonempty  boxes in those
  grids. Each figure represents a different configuration type with
  respect to the appearance of some $4$-length pattern. The dots and
  the numbers indicate possible splittings of the $4$ legs  
of $\pi$.
Figure (E) represents the pattern $(4,2,1,3)$ and all others represent the pattern $(3,2,1,4)$.
The sizes of green boxes in the figures are not representative and are not drawn to scale.}
\label{fig:1}
\end{figure}

For $\pi \in \mathcal S_k$, a $\pi$-appearance in $G_n$ implies that the $k$ points corresponding to such a
$\pi$-appearance are in $i \leq k$ distinct components  in $G_{m,m}$,
where the $j$th component $C_j$ may contain~$b_j$ boxes each
containing at least one point of the corresponding
$\pi$-appearance. 
We refer to the
$\pi$-values in the corresponding boxes of the components as \emph{legs}.  For example, for
$\pi = (3,2,1,4)$, the $\pi$-appearance shown in \cref{fig:1}(B) is
contained in two boxes that  share the same layer, and hence form one
component. The left box contains the
$3,2$ legs of the $\pi$-appearance and other contains the $1,4$ legs.
  A different $1$-component  $2$-boxed 
 appearance in the same two boxes has $3$ appearing in
$B_1$ and all the other legs in $B_2$ as in \cref{fig:1}(C).
It need not be the case that every pair of boxes in a single connected component are directly-connected
as illustrated in~\cref{fig:1}(J) and \cref{fig:1}(K).

  Examples for $\pi=(3,2,1,4)$-appearances with two components
  $C_1,C_2$ are illustrated in \cref{fig:1}(F) and \cref{fig:1}(H). In
  the first, $C_1,C_2$ contain $2$ boxes each, where $C_1$ contains the
  $(3,4)$ legs of the appearance, each in one box, and $C_2$
  contains the $(1,2)$ legs. In the second, each component is
  $1$-boxed, where the first contains the $(3,2,1)$-legs and the other
  contains the $4$-leg of the appearance. 
  \cref{fig:1}(A) contains a $(3,2,1,4)$-appearance in $4$
  components. Some other possible appearances with $1$ component 
 and $3$ components  are illustrated in \cref{fig:1}(B), \cref{fig:1}(C), 
 \cref{fig:1}(D) and \cref{fig:1}(G).

 To sum up, each $\pi$-appearance in
$G_n$ defines an arrangement of nonempty boxes in $G_{m,m}$ that contain the legs of that appearance.  This arrangement is defined by the relative order of the
layers and stripes among the boxes, and has at most $k$
components. 
Such a box-arrangement that can contain the legs of a $\pi$-appearance is called a
\emph{configuration}.
Note that there may be many different
$\pi$-appearances in distinct boxes, all having the same 
  configuration $\mathcal C$. 
  Namely, in which, the arrangements of the boxes
in terms of the relative order of layers and stripes are identical.
So, every set of $\ell \leq k$ points in the $k \times k$ grid defines a configuration and two such sets represent the same configuration
if they are order-isomorphic with respect to the grid order. For instance, the sets of points $\{(1,1), (2,1), (3,3)\}$ and $\{(1,1), (2,1), (3,4)\}$ represent the same configuration.
An actual set of boxes in $G_{m,m}$ forming a specific type of configuration
is referred to as a \emph{copy} of that configuration.

For $\pi \in \mathcal S_k$, let $c(k)$ be the number of all possible configurations that are consistent
with a $\pi$-appearance. 
For any fixed $\pi, $ the number $c(k)$ of
distinct types of configurations is upper bounded
in the following observation.

\begin{observation}
  \label{obs:c7}
  $c(k) = 2^{O(k \log k)}$ 
\end{observation}
 \begin{proof}
 The total number of possible configurations is upper bounded by the
 number of ways to select at most $k$ points from a $k \times k$
 grid. This latter quantity is equal to $\sum_{i \in [k]} \binom{k^2}{i}$, which is at most $2^{O(k \log k)}$. 
 \end{proof}

A configuration $\mathcal{C}$ does not fully specify the way in which a $\pi$-appearance can be present. 
It is necessary to also specify the way the $k$ legs of the $\pi$-appearance are partitioned
among the boxes in a copy of $\mathcal{C}$. 
Let $\mathcal{B}$ denote a set of boxes forming the configuration $\mathcal{C}$.
Let $\phi: [k] \to \mathcal{B}$ denote the mapping of the legs of the $\pi$-appearance to boxes in
$\mathcal{B}$, where $\phi(j), j \in [k]$ denotes the box in $\mathcal{B}$ containing the $j$-th leg of the $\pi$-appearance.
We say that the copy of $\mathcal{C}$ formed by the boxes in $\mathcal{B}$ contains a $\phi$-legged $\pi$-appearance.

A configuration $\cC$ in which the boxes form $p \geq 2$
components, and that is consistent with a $\pi$-appearance, defines
$\nu_1, \ldots ,\nu_p$-appearances, respectively, in the $p$ components
of $\cC$, where $\nu_j$ for $j \in [p]$ is the subpermutation of $\pi$ that is defined by
the restriction of $\pi$ to the $j$-th component.  In addition, $\cC$
defines  the corresponding 
mappings $\phi_j, j=1, \ldots p,$ of the corresponding legs of each $\nu_j$ to the corresponding
boxes in the $j$th component. For example,
consider $\pi = (3,2,1,4)$ and the box arrangement shown in \cref{fig:1}(F). That arrangement has two connected
components: one that contains $B_1,B_4$ and the other that contains
$B_2,B_3$, where we number the boxes from left to right (by increasing stripe
order). Further, the (only) consistent partition of the legs of $\pi$ into these
boxes is  $\pi(i) \in B_i, ~ i\in [4]$. In particular, it means
that the component formed by $B_1, B_4$ contains the $3,4$ legs of $\pi$ and the
component formed by $B_2, B_3$ contains the $2,1$ legs of $\pi$.
Thus, in terms of the  discussion above, the component formed by $B_1, B_4$
has a $\nu_1 = (1,2)$-appearance (corresponding to the $3,4$ legs of
$\pi$), with leg mapping $\phi_1$ mapping the $1$-leg  into $B_1$ and
the $2$-leg  into 
$B_4$. Similarly,  the component formed by $B_2, B_3$ has a $\nu_2=(2,1)$-appearance
(corresponding to the $2,1$ legs of $\pi$) with corresponding leg
mapping~$\phi_2$ that maps the $2$-leg into $B_2$ and the $1$-leg into $B_3$. Note
that the converse is also true:  every $\nu_1$-appearance in
the component $B_1 \cup B_4$, with a leg-mapping $\phi_1$ (that is, in which the $1,2$ legs are in
$B_1,B_4$ respectively),  in addition to a $\nu_2$-appearance in
$B_2 \cup B_3$ with the leg-mapping $\phi_2$,  results in a
$\pi$-appearance in $G_{m,m}$.  

This leads to the crucial observation that if $\pi$ defines the corresponding
$\nu_1, \ldots , \nu_p$ appearances in the $p$ components of the
configuration $\cC$,  then, \emph{any} $\nu_1, \ldots ,\nu_p$-appearances
in the $p$ components of any copy of $\cC$ with consistent
leg-mappings is  a $\pi$-appearance in
$\cC$. This is formally stated below.

\begin{definition}
  \label{def:restricted_appearance}
  Let $\nu \in \mathcal S_r$. 
  Let $B_1, \ldots ,B_p$ 
  be a set of boxes forming
  {\rm one} component $C$
  and  $\phi : [r] \mapsto \{B_1, \ldots, B_p\}$
  be an arbitrary mapping of the legs of a $\nu$-appearance
  to boxes.  We say that $C$ has a $\phi$-legged $\nu$-appearance
  if there is a
  $\nu$-appearance in $\bigcup_{j = 1}^p B_j$ in which for each $i \in [r]$,
  the $i$-th leg of $\nu$ appears in the box $B_{\phi(i)}$. 
\end{definition}


\begin{observation}
  \label{obs:3}
  Let $\pi \in \mathcal{S}_k$ and assume that there exists a
  $\pi$-appearance in $G_n$ that, in the grid of boxes $G_{m,m}$, forms a configuration $\cC$ that contains 
  $t$ components $\mathcal{C}_1, \ldots ,\cC_t$.
  Let  the restriction of this
  $\pi$-appearance to $\cC_1, \ldots ,\cC_t$ define the permutation patterns $\nu_1, \ldots \nu_t$
  with leg mappings $\phi_1, \ldots \phi_t$, respectively.
  Then any collection $\{C'_j: j \in [t]\}$ such that $C'_j$ is a configuration copy of $\cC_j$ and $\bigcup_{j=1}^{t} C'_j$ is a copy of $\cC$, along with 
  $\phi_j$-legged $\nu_j$-appearances in $C'_j$ for each $j \in [t]$  
 defines a
  $\pi$-appearance in $\bigcup_{j=1}^{t} C'_j$. 
\end{observation}
\begin{proof}
Since $\bigcup_{j \in [t]} C'_j$ form a copy of the configuration $\cC$, for any two boxes $B_1$ and $B_2$ belonging to $\bigcup_{j \in [t]} C'_j$, their relative position in the grid $G_{m,m}$ is identical to the relative position of the corresponding boxes in $\cC$. 
For $a,b \in [k]$ such that $a < b$, consider the $a$-th and $b$-th leg in the order from left to right along the grid, in the union of $\phi_j$-legged $\nu_j$-appearances in $C'_j$ for $j \in [t]$.
By the above statement and by virtue of the leg mappings $\phi_j, j \in [t]$, the relative values of the $a$-th and $b$-th legs in the aforementioned union of appearances is identical to the relative values of the $a$-th and $b$-th legs in the $\pi$-appearance occurring according to the configuration $\cC$. Therefore, the union of $\phi_j$-legged $\nu_j$-appearances in $C'_j$ for $j \in [t]$ defines a $\pi$-appearance.
\end{proof}

\subsection{Erasure-resilient testing}

Erasure-resilient (ER) testing, introduced by Dixit, Raskhodnikova, Thakurta and Varma~\cite{DixitRTV18}, is a generalization of property testing. 
In this model, algorithms get oracle access to functions for which the values of at most $\alpha$ fraction of the points in the domain are erased by an adversary, for $\alpha \in [0,1)$. 

For $f:[n]\to \R$ let $\NE(f)$ be the
\emph{nonerased} values of $f$.
The parameter $\alpha$ is given as an input to the algorithms, but, they do not know $\NE(f)$.
On querying a point, the algorithm receives the function value if the point is \emph{nonerased}, and a special symbol otherwise.

\begin{definition}[One-sided error erasure-resilient tester for ${\mathcal P}_\pi$]~\label{def:er-testing}
For $\eps \in (0,1), \alpha \in [0,1),$
 an $\alpha$-erasure-resilient ($\alpha$-ER) $\eps$-tester for
$\mathcal{P}_\pi$ is a randomized algorithm that on oracle access to a function $f:[n] \to \mathbb{R}$, 
\textbf{accepts}, with probability $1$, if $f|_{\NE(f)}$ is $\pi$-free, and 
 \textbf{rejects}, with probability at least $2/3$, if there is a
 matching of size $\eps n/k$ of $\pi$-appearances in $\NE(f)$.
\end{definition}    

We point out that the definition in~\cite{DixitRTV18} is for any
property and for two-sided error testing as well.

Dixit et al.~\cite{DixitRTV18} give a one-sided error $\alpha$-ER $\eps$-tester for
monotonicity of functions $f:[n] \to \mathbb{R}$ with query complexity
$O(\frac{\log n}{\eps})$ that works for any constants
$\alpha, \eps \in [0,1)$. 
It can be observed that the $\polylog n$-query
one-sided error tester for $\nu$-freeness of~\cite{NewmanRRS19}, for any
 $\nu \in \mathcal{S}_3$, is also ER.

As part of our
algorithm for testing $\pi$-freeness for  $\pi \in \mathcal{S}_k$
for $k  \geq 4$, we call testers for smaller subpatterns on subregions of the grid $G_n$ which may be defined by, say, $\bx(S,I)$ for some $S \subseteq [n], I \subseteq R(f)$. 
In this case, the only access to points in $\bx(S,I)$ is by sampling indices from $S$ and checking whether their values fall in $I$. 
If the values do not fall in $I$, we can treat them as erasures. 
Given the promise that the number of points falling in $\bx(S,I)$ is a constant fraction of $|S|$, we can 
simply run ER testers on $f|_S$ to test for these smaller subpatterns. 

\section{High-level description of the basic algorithm for
  \texorpdfstring{$\pi \in {\mathcal S}_4$}{testing freeness of a pattern of size 4}}\label{sec:top-level}

In this section, we give a high-level description of most of the ideas used in the design of our $\pi$-freeness tester of query complexity
 $\tilde{O}(n^{o(1)})$. 
We first describe the ideas behind a $\tilde{O}(\sqrt{n})$-query $\eps$-tester for $\pi$-freeness of functions $f:[n] \to \R$, where $\pi \in \mathcal{S}_4$ and $\eps \in (0,1)$.
At the end of this section, we briefly touch upon how to generalize
these ideas to obtain the query complexity of $\tilde{O}(n^{o(1)})$
for constant-length permutations of length at least $4$. For simplicity, we assume in what follows that the input function
$f: [n] \to \R$ is one-to-one.  The algorithm for functions that are not
one-to-one  differs in a few places and these are explained in~\cref{sec:proof-correct1}.

For the purposes of this high-level description, we fix the forbidden
permutation  $\pi
= (3,2,1,4)$. The same algorithm works for any
$\pi \in \mathcal S_4$.
We view $f$ as an
(implicitly given) $n \times |R(f)|$ grid $G_n$ consisting of points
$(i,f(i))$ for $i \in [n]$, where, in
particular, $R(f)$ is neither known nor bounded.
Our first goal is to approximate $G_n$ by a coarse grid of boxes
$G_{m',m'}$, where $m= \sqrt{n}$ and $m' = \Theta(m)$.
This is done by first querying $f$ on $\tilde{\Theta}(m)$ independently sampled and uniformly random indices,
upon which we obtain a partition $\mathcal L$ of $R(f)$ into $m'$ horizontal layers, corresponding to value intervals $\{I_j\}_{j \in [m']}$. 
 Then,
we partition
the index set $[n]$ into
$m'$ contiguous intervals $\{S_i\}_{i\in [m']}$ of equal sizes. 
This results in
a grid $G_{m',m'}$, 
where a box $\bx(S_i,I_j),i,j \in [m']$ is tagged as \emph{nonempty} if it has at least one sampled point. A box is tagged as \emph{dense} if it
contains $\Omega_{\eps}(1)$-fraction of the sampled points in its stripe.
All of the above takes $\tilde{O}_{\eps}(m)=\tilde{O}_{\eps}(\sqrt{n})$ queries.
The following properties are satisfied with high probability:
\begin{itemize}
\item Each layer, that is $\bx([n],I_j), ~ j \in [m'],$ has
  approximately the same number of points from $G_n$.
  \item It is either the case that the dense boxes contain all but an insignificant fraction of
    the points in $G_n$, or the total number of nonempty boxes is larger
    than $m' \log n$.
\end{itemize}

Next, we use the following lemma of Marcus and Tardos.
\begin{lemma}[\cite{MarcusT04}]\label{clm:MarcusTardos}
For any $\pi \in \mathcal S_k,~ k \in \N$, there is a constant $\kappa(k) \in \N$ such that for any $r \in \N$, if a grid $G_{r,r}$
contains at least $\kappa(k) \cdot  r$ marked points, then it contains a
$\pi$-appearance among the marked points. 
\end{lemma}

Let $\kappa= \kappa(4)$. Using~\cref{clm:MarcusTardos}, we may assume that there are at most $\kappa \cdot m'$ nonempty boxes in $G_{m',m'}$, as otherwise, we already would have found a
$\pi$-appearance in $G_{m',m'},$ which by~\cref{obs:homo}, implies a
$\pi$-appearance in $G_n$ and in $f$ as well. Hence, as a result
of the gridding, if we do not see a $\pi$-appearance among the
sampled points, the second item above implies that there are
$\Theta(m')$ dense boxes in $G_{m',m'}$ and that these boxes cover all
but an insignificant  
fraction of the points of $G_n$. 

An averaging argument implies that, for an appropriate value $d = d(\eps)$, only a small fraction (depending on $\eps$) of layers (or stripes)
contain more than $d$ nonempty boxes. Therefore, since the grid $G_n$ is $\eps$-far from being $\pi$-free,
the restriction of $G_n$ to the layers and stripes that  contain at most $d$ boxes each, is also $\eps'$-far from $\pi$-free for
a large enough $\eps' < \eps$. This implies that $G_n$
restricted to the points in dense boxes that belong to layers and stripes
containing at most $d$ dense boxes each, has a matching $M$ of
$\pi$-appearances of size at least $\eps' n/4$. We assume in what
follows that this is indeed the situation. 

An important note at this point is that every dense box
$B$ is contained in $O(d^3)$ many copies of
$1$-component configurations with at most $4$ dense boxes.
This implies that there are $O(d^3m)$ such copies of $1$-component configurations in $G_{m',m'}$.

Recall that every $\pi$-appearance in $M$ defines a configuration
of at most $4$ components in $G_{m',m'}$. 
Hence,  the matching $M$ of size $|M| = \Omega_{\eps}(n)$
 can be partitioned into $4$ sub-matchings $M = M_1 \cup M_2 \cup M_3
 \cup M_4$, where $M_i, ~i =1, \ldots ,4$ consists of the $\pi$-appearances participating in
 configurations having exactly $i$ components. Since $|M| = \Omega_{\eps}(n)$
 it follows that at least
 one of $M_i, ~ i=1,2,3,4$ is of linear size. Now, any $\pi$-appearance in $M_4$ is an appearance in $4$ distinct dense boxes in
 $G_{m',m'}$, where no two share a layer or a stripe. In that case, such
 an appearance can be directly detected from the tagged $G_{m',m'}$ with no further queries.

The description of the rest of the algorithm  can be viewed as a treatment of several
independent cases regarding which one among the constantly many configuration types
contributes the larger mass out of the $\Omega_\eps(n)$
$\pi$-appearances in $M_1 \cup M_2 \cup M_3$. There are only two
significant cases, but to
enhance understanding, we split these two cases into the more natural larger
number of cases, and observe at the end  that most cases can
 be
treated conceptually  in the same way.

\paragraph{Case 1:} Let $|M_1| \geq \eps' n/3$, and let a constant fraction of the $\pi$-appearances in $M_1$
be in a single-box component. Then, on average,  a dense box,
out of the $\Theta(m')$ dense boxes, is expected to contain  at least
$\Theta_{\eps}(n/m') = \Theta_{\eps}(m') = \Theta_{\eps}(\sqrt{n})$ many $\pi$-appearances. Thus a random dense box $B$ is likely to have 
 $\Theta_{\eps}(\sqrt{n})$ many $\pi$-appearances, and hence,  making queries to
all points of such a box will enable us to find one such $\pi$-appearance. This takes
 an additional $n/m' = \Theta(\sqrt{n})$ queries, which is within the
 query budget.

Next, consider the case that a constant fraction of the
$\pi$-appearances in $M_1$ belong to a configuration $\cC$ that has more than
one dense box (but only one connected component). An example of such a situation would be~\cref{fig:1}(J). By a similar argument, a random dense
box is expected to participate in at
least $\Theta_{\eps}(n/m)$ many $\pi$-appearances of copies of configuration-type~$\cC$.
Since each
dense box is part of at most $O(d^3)$ (constantly many) connected components of at most
$4$ dense boxes,  sampling a random dense box $B$ and querying all the indices in each of
the components that contain at most $4$ dense boxes and involve~$B$, is likely to find a $\pi$-appearance with
high probability. 
Each connected component is over at most $4n/m'$ indices, resulting in
$O(n/m)$ queries.

\paragraph{Case 2:}
$|M_3| \geq \eps' n/3$, and 
assume first that a constant
fraction of the members in $M_3$ belong to copies of a
configuration $\cC$ of $3$ components $B_1,B_2,B_3$, 
where each one is a single box. Since the boxes $B_1,B_2,B_3$ belong to different components, no two of them share a layer or a stripe. For our current working example,  $\pi
= (3,2,1,4)$, assume further that $B_1$ contains the $3,2$ legs of a
$\pi$-appearance and $B_2,B_3$ contain its 
$1$ and $4$ legs, respectively (see~\cref{fig:1}(G) for an example). In this case $B_1$ is not
$(2,1)$-free (as $B_1$ contains the $(3,2)$-subpattern of $\pi$). 

By an averaging argument, it follows that there is a dense box $B$
for which: (a) $B$ is far from $(2,1)$-free, and (b) there are
corresponding dense boxes $B_2,B_3$ that, together with $B$, form a copy of the
configuration $\cC$. Now, a test follows
 easily. We test every dense box for $(2,1)$-freeness, which can be
 done in $O(\log n)$ queries per box, and hence in $\tilde{O}(m)$ in
 total. Then, by the guarantee above we will find the corresponding
 $B,B_2$ and $B_3$ and a $\pi$-appearance in it (by~\cref{obs:3} with the trivial
 mapping).

  A similar argument holds for a $3$-component configuration $\cC'$ in
  which one component contains more than one box.
  Let $\cC'$ consist of two single-box components and a two-boxed component, as in~\cref{fig:1}(D). 
In this case, a similar averaging argument shows the existence of a
dense box $B$ for which (a) there is a dense box $B'$ forming a
component $D$ with $B$, and dense boxes $B_2,B_3$ such that $D,B_2,B_3$ jointly form a copy of $\cC'$, and (b)
there are $\Omega_{\eps}(n/m) = \Omega_{\eps}(\sqrt{n})$  $\phi$-legged $(2,1)$-appearances in
$D$, where $\phi$ is such that the $2$-leg maps to the upper box in $D$ and the $1$-leg maps to the lower box in $D$. 
Hence, the test is similar to the simpler
case above. We test for every dense box $B$ and every way to extend it
into a component of two boxes by adding a box $B'$ (a constant number
of ways) such that $D = (B,B')$ contains a $\phi$-legged  $(2,1)$-appearance.  This
again can be done using $O(\log n)$ queries per component copy $D$. Once this is done,
finding $D,B_2,B_3$ that form a copy of $\cC'$ results in a
$\pi$-appearance by~\cref{obs:3}.


\paragraph{Case 3:}
Assume now that $|M_2| \geq \eps' n/3$, and  that the
corresponding 
configurations of the $\pi$-appearances in $M_2$ contain two single-box components $B_1,B_2$,
where $B_1$ holds the first $3$ legs of $\pi$ and $B_2$ holds the
$4$-th leg. E.g., For  $\pi=(3,2,1,4)$, the configuration $\cC$
contains two boxes $B_1,B_2$ where $B_1$
contains the subpattern $(3,2,1)$ and $B_2$ is any nonempty box
such that $B_1 < B_2$, (see~\cref{fig:1}(H) for an illustration).
 An averaging argument, as made in Case 2,  shows that there is
a
dense box $B_1$ for which (a) $B_1$ is far from $(3,2,1)$-free, and (b) there is
a corresponding dense box $B_2$ that, together with $B_1$, forms a copy of the
configuration $\cC$.
    This suggests a test that is 
    conceptually similar to the test in Cases 1 and 2. We test each
    box for being $(3,2,1)$-free. This can be done in $O(\polylog n)$
    queries (e.g., \cite{Ben-EliezerCLW19}). Then once finding a $(3,2,1)$ in $B_1$ for
which (a) and (b) hold, $B_1 \cup B_2$ contains a $\pi$-appearance.

We note here
that for the example above, we ended by testing for $(3,2,1)$-freeness
which is relatively easy. For a different configuration or $\pi$, we might need
to test  $B_1$ for a different $\nu \in \mathcal S_3$, but this can be
done for any $\nu \in \mathcal{S}_3$ using $O(\polylog
n)$ queries~\cite{NewmanRRS19}. Hence the same argument and complexity guarantee
hold for any $2$-component configuration $\cC$ as above.

\paragraph{Case 4:}
A  more complicated situation  arises when  $|M_2| \geq \eps' n/3$, and the
corresponding 
configurations of the $\pi$-appearances in $M_2$ are formed of two
components $D,B$, with $D$ holding $3$ legs of $\pi$ in  $2$ or $3$
boxes (rather than in one box as in Case 3). E.g., 
$\pi=(4,2,1,3)$, and the configuration $\cC$ as illustrated  in~\cref{fig:1}(E).

By a similar averaging argument to that made in Case 2, it follows that
there is a 
dense box~$B_1$ for which (a) there are dense boxes $B_2,B_3$ forming
a copy  $D'$ of $D$ with $B_1$, and a dense box~$B$ such that the
configuration formed by $D',B$ is a copy of $\cC$, and (b)
there are $\Omega_{\eps}(n/m) = \Omega_{\eps}(\sqrt{n})$  $\phi$-legged $(3,2,1)$-appearances in
$D'$, where $\phi$ is consistent with the leg mapping that is induced by
the  configuration $\cC$.  This implies a conceptually similar test to
that of 
the simpler Case 3 above -  we test each of the $O(m)$ components $D$
for $(3,2,1)$-freeness, and then with the existence of the
corresponding box $B$ we find a $\pi$-appearance.  However, this is
not perfectly accurate:  the algorithm for finding $\nu=(3,2,1)$ in $D'$, although 
efficient, might find a $(3,2,1)$-appearance where the $3$ legs appear
in $B_1$ or in $B_1 \cup B_2$.  But this does not extend with $B$ to form
a $\pi$-appearance, as the leg mapping is not consistent with 
the one that is induced by $\cC$.   Namely, unlike before, we do not only need
to find a $\nu$-appearance in $D$ but rather a $\phi$-legged 
$\nu$-appearance  with respect to a fixed mapping $\phi$ (that in this case maps
each leg to a different box in the component $D'$).

There are several ways to cope with this extra restriction.
 For the current description of a basic $\tilde{O}(\sqrt{n})$ algorithm, it is enough to sample a
constant number of copies of the component $D$ and do the test for
$\phi$-legged $\nu$-appearance in each. But, since each copy $D'$ is of size 
$O(\sqrt{n})$ we can afford to query all indices in the domain of
$D'$.

To resolve the problem in the general setting, we need to efficiently detect $\phi$-legged $\nu$-appearances
in multi-boxed components. This, however, we currently do not know how to do. Instead, we design a test that
either finds a $\phi$-legged $\nu$-appearance, or finds the
original $\pi$-appearance. This is done  using the algorithm
$\mathsf{AlgTest}_{\pi}(\nu,\phi,D, m,\eps)$ that will be described in~\cref{sec:generalizedtesting}.

\paragraph{Case 5:}
 The last case that we did not consider yet is when most of the
 $\pi$-appearances are in a configuration containing more than one
 component, with at least two components containing two (or more) legs
 each.  For $\pi \in \mathcal S_4$ the only such case is when the
 configuration $\cC$ contains exactly two components, each containing
 exactly two legs of $\pi$.   Returning to our working example with
 $\pi=(3,2,1,4),$ such an example is depicted in~\cref{fig:1}(F).
For the explanation below, we will discuss the case that the 
configuration $\cC$ is as in~\cref{fig:1}(F). Namely, it contains
components $D_1$ that is above $D_2$, with two boxes each $D_1 = \{B_1,B_4\}$ and  $D_2
= \{B_2,B_3\}$, and so that every box contains exactly one leg of
$\pi$ (boxes are numbered by order from left to right in $G_{m',m'}$).
Our goal is to find two
copies $D_1',D_2'$ of the components $D_1,D_2$ respectively,  that form a copy of $\cC$, and to find a
$\phi_1$-legged appearance of $(1,2)$ in $D_1'$, and a $\phi_2$-legged appearance of
$(2,1)$ in $D_2'$, so that
these two appearances will together form a $\pi$-appearance.

Indeed, an averaging argument shows that there are $D_1',D_2'$ as above,
with $D_i'$ containing $\Omega_{\eps}(n/m)$ $\phi_i$-legged appearances of
$\nu_i$ for $i=1,2$.  However, 
 we do
not know whether sampling a pair  $D_1',D_2'$ in some way, will
result in such a good pair. Rather, we are only assured of the existence
of only one such pair!  Hence, in this case we need to test {\em every}
component copy $D'$ of the appropriate type, for every $\nu \in \mathcal S_2$, and for every leg
mapping $\phi$, for a $\phi$-legged $\nu$-appearance in $D'$ in order
to find such an asserted pair of components.  Such
restricted $\nu$-appearances can be tested in $O(\log n)$ queries per
component. Since the number of two-boxed component copies where both boxes belong to the same layer is $O(m)$, this step takes $\tilde{O}(m)$ queries in total.

The same argument holds for any $\pi \in \mathcal S_4,$ and for every 
configuration that is consistent with Case 5.

\paragraph{Concluding remarks}
\begin{itemize}

\item At some places in the algorithm above, we had to
test for $\nu$-appearances (or restricted $\nu$-appearances) in
`dense' subgrids of $G_n$. For this, we need all
our algorithms to be ER, which will be implicitly clear from the
description.  We also need to take care of reducing the total error
when we run a non-constant number of tests, or want to guarantee a
large success probability for a large number of events - this is done
by a trivial amplification that results in a multiplicative $\polylog
n$ factor.

\item  In Case 1, we reduced the problem of finding a
 $\pi$-appearance in $G_n$ that is assumed to be $\eps$-far from
 $\pi$-free, to the same problem on a subrange of the indices (formed
 by a small component) of size
 $\Theta(n/m)$ (with a smaller but constant distance parameter $\eps' <
 \eps$). For the setting of $m=\sqrt{n},$ solving the problem on
 the reduced domain was trivially done by  querying all
 indices in the subrange. In the general algorithm, where our goal is a query complexity of $n^{o(1)}$, we set $m= n^{\delta}$ for an appropriately small
 $\delta$ and apply self-recursion in Case 1.

\item In Case 5, we had to test for $\nu$-freeness (or
for restricted $\pi$-appearances) for $\nu \in \mathcal S_2$ for {\em every}
small component of size $\Theta(n/m)$ in $G_{m',m'}$.  This entails a collection of
$O(m)$ tests, where we want to assign a large success probability to each one of them. 
We also need to guarantee a large success probability to correctly tagging
each of the $\Theta(m^2)$ boxes as part of the layering procedure. 
A similar need will also arise in the general
algorithm.  We amplify the success probability by multiplying our
number of queries by $\log^2 n$ which will imply less than $1/n^{\Omega(\log n)}$
failure probability for each individual event in such collection.  We
will not comment more on this point, and assume implicitly that in all
such places, all needed events occur w.h.p.  


\item In Cases 2, 3, 4 we end up testing $\nu$-freeness for
 $\nu \in \mathcal S_2 \cup \mathcal S_3$ in dense boxes, or
 $\phi$-legged $\nu$-freeness of such $\nu$ in components of multiple
 dense boxes. An averaging argument shows that this can simply be done by sampling one box or
 component, and making queries to all indices therein. 

 Case 5 is different: here, sampling a small number of components does
 not guarantee an expected large number of the corresponding
 appearances. This is the reason that we
 need to test {\em all} components with at most $2$ dense boxes, for
 $\phi$-legged $\nu$-freeness, and for every $\nu \in \mathcal{S}_2$ and leg mapping $\phi$. Algorithm
 $\mathsf{AlgTest}_{\pi}(\nu,\phi,D, m,\eps)$ can do this for any
 $\nu \in \mathcal S_2 \cup \mathcal S_3$ in $n^{\delta}$ queries for
 an arbitrarily small constant $\delta$. Since we have to do it in Case
 5, we may do the same in cases 2, 3, 4 as well! As a result, the
 algorithm above will contain only two cases: Case 1 where we reduce
 the problem to the same problem but on a smaller domain, and the
 new Case 2 where we test {\em every} small component for
 $\phi$-legged $\nu$-appearance for every
 $\nu \in \mathcal S_2 \cup \mathcal S_3$ and every leg mapping $\phi$ -- namely a
 case in which we reduce the problem to testing (restricted
 appearances) for smaller patterns.

\item In view of the comment above, the idea behind improving the
complexity to $n^{\delta}$ for constant $0< \delta <
1$ is obvious:  Choosing $m= n^{\delta/2}$ will result in an $m \times m$
grid, where Layering can be done in
$\tilde{O}(n^{\delta/2})$ queries. Then, Case 2 will be done in  an additional $n^{\delta}$
queries by 
setting a query complexity for $\mathsf{AlgTest}_{\pi}(\nu,\phi,D, m,\eps)$
to be $n^{\delta/2}$ per component. The self-recursion in Case 1 will result in the same
problem over a range of $n/m$. For the fixed $m=n^{\delta/2},$ this will result in
a recursion depth of $2/\delta$, after which the domain size will drop
down to $m$ and allow making queries to all corresponding
indices. This results in a total of $\tilde{O}(n^\delta)$
queries, including the amplification needed to account for the
accumulation of errors and deterioration of the distance parameter at
lower recursion levels.

\item \textbf{Generalized testing and testing beyond $k=4$}. Applying the same ideas to $\pi \in
\mathcal S_k,~ k\geq 5$ works essentially the same way,  provided we can
test for $\phi$-legged $\nu$-freeness of $\nu \in \mathcal S_r$ for $r
< k$. This we know how to do for $\nu \in \mathcal S_2$ but not beyond.
For $r=2$, testing $\phi$-legged $\nu$-freeness of
$\nu \in \mathcal S_2$ is simpler than testing monotonicity for
nontrivial $\phi$, and is equivalent to testing monotonicity when
testing is done in a one-boxed component.  Hence, this can be done in
$O(\log n)$ queries. 
For $r \geq 3$ the exact
complexity is currently not know.

 In particular, one difficulty is that after gridding, a
superlinear number of nonempty boxes does not guarantee such
appearance, 
as Lemma~\ref{clm:MarcusTardos} does not apply. For example, for even~$r$,
consider the grid $[r] \times [r]$ all of whose points in the top left quarter $\{1,\dots,r/2\}\times \{r/2 + 1,\linebreak[4]\dots,r\}$ and right bottom quarter $\{r/2 + 1,\dots,r\} \times \{1,\dots,r/2\}$ are marked. There are no restricted $(1,2)$-appearances among the marked points where the $1$ leg is from the right half and the $2$-leg is from the left half, despite there being $\Omega(r^2)$ points.
However,  for our
goal of testing $\pi$-freeness for $\pi \in \mathcal S_k,$  we can relax the task of finding
$\phi$-legged $\nu$-freeness of $\nu \in \mathcal S_r, r \leq k$ to the following
problem which we call ``generalized-testing $\nu$ w.r.t.\ $\pi$'', denoted $\mathsf{GeneralizedTesting}_{\pi}(\nu, D, \phi)$:  The inputs are a permutation $\nu \in \mathcal{S}_r$,
a component $D$, and a leg mapping $\phi$. Our goal is to find either a
$\phi$-legged $\nu$-appearance \textbf{OR} a $\pi$-appearance in $D$.
The way we solve this generalized problem is very similar, conceptually, 
to the way we solve the unrestricted problem; we decompose $D$  into an $m \times
 m$ grid of subboxes,$D_{m,m}$, by 
 performing gridding of $D.$  Then, we either find $\pi$ in
 $D_{m,m}$, or, 
using Lemma~\ref{clm:MarcusTardos}, conclude that
there are only linearly many dense subboxes in $D_{m,m}$. At that point, we
find a $\phi$-legged $\nu$-appearance by reducing it to the same problem
of a 
$\phi'$-legged $\nu'$-freeness of smaller $\nu' \in \mathcal S_{r'}, ~ r' < r,$
or, self-reducing the problem for finding $\phi$-legged $\nu$-appearance
but in a sub-component $D'$ whose size is a factor $m$ smaller than that of the size of $D$. This is done in a
similar way to what is described above in Case 1.

In summary, the algorithm for  $\mathsf{GeneralizedTesting}_{\pi}(\nu)$ is
very similar to the algorithm for testing $\pi$-freeness, with the same two
cases, where Case 2 becomes recursion to finding appearances of a smaller permutation, and where the base case is for permutations of length $2$. As we show in Section
\ref{sec:generalizedtesting}, formally, 
$\mathsf{GeneralizedTesting}_{\pi}(\nu)$  strictly generalizes testing
$\pi$-freeness, and hence, the formal algorithm for testing
$\pi$-freeness will be a special case of $\mathsf{GeneralizedTesting}_{\pi}(\nu)$.

\end{itemize}

\section{Gridding}\label{sec:layering}

In this section, we describe an algorithm that we call
Gridding (\cref{proc:preprocess}), which is a common subroutine to all our algorithms.
The output of Gridding, given oracle access to the function $f:[n] \to \R$ and a parameter $m \leq n$, is an $m \times m$ grid of boxes
that partitions either the grid $G_n$ defined by $f$ or a region inside of it into boxes,
with the property that the density of each box, which we define below, is well controlled.

\begin{definition}[Density of a box]\label{def:density}
Consider index and value subsets $S \subseteq [n]$ and $I \subseteq R(f)$, respectively.
The density of $\bx(S,I)$, denoted by $\den(S,I)$, is 
the number of points in $\bx(S,I)$ normalized by its size $|S|$. 
\end{definition}

\begin{definition}[Nice partition of a box]\label{def:nice-partition}
For index and value sets $S \subseteq [n]$ and $I \subseteq R(f)$ and parameter $m \leq n$, we say that
${\mathcal I} = \{I_1, I_2, \dots, I_{m'}\}$ forms a \textsf{nice} $m$-partition of $\bx(S,I)$ if:
\begin{itemize}
\item $m' \leq 2m $,
\item $I_1, \ldots ,I_{m'}$ are pairwise disjoint, and $\bigcup_{j
    \in [m']} I_j = I$. In particular, 
  the largest value in $I_j$ is less than the  smallest value in $I_{j'}$ for $j < j'$.
\item for $j \in [m']$, either $\den(S,I_j) < \frac{4}{m}$ $\mathrm{OR}$ $I_j$ contains exactly one value and
 is such that $\den(S,I_j) \geq \frac{1}{2m}$. 
  In the first case, we say that $\bx(S,I_j)$ is a
  \emph{single-valued layer} of $\bx(S,I)$, and in the second case, we say that $\bx(S,I_j)$ is a \emph{multi-valued layer} of $\bx(S,I)$.
\end{itemize}
\end{definition}

\subsection{Layering}
The main part of Gridding is an algorithm
Layering which is described in~\cref{proc:layering}. A similar algorithm was used by Newman and Varma~\cite{NV20} for estimating the length of the longest increasing subsequence in an array.
Layering$(S,I,m)$, given $S \subseteq [n], I \subseteq R(f), m \leq n$ as inputs, and outputs, with probability at least
$1 - 1/n^{\Omega(\log n)}$, a set  $\mathcal I$ of intervals that is  a nice
$m$-partition of  $\bx(S,I)$.  It works by  sampling
$\tilde{O}(m)$ points from $\bx(S,I)$ and outputs the set $\mathcal{I}$ based on these samples. 
Note that both the sets $S$ and $I$ are either contiguous index/value intervals themselves or a disjoint union of at most $k$ such contiguous intervals. 
Additionally, we always apply the algorithm Layering to boxes of
density $\Omega(1/\log n)$.

  \begin{algorithm}
  \caption{Layering($S,I, m$)}
  \label{proc:layering}
\begin{algorithmic}[1]
\State Sample a set of $m \log^4 n$  
indices from $S$ uniformly and independently at random.
\State Let $U$ denote the multiset of points in the sample that belong to $\bx(S,I)$ and let $u$ denote the cardinality of $U$ including multiplicities. If $u < m \log^2 n$, then \textbf{FAIL}. \label{stp:layering-fail}

\State We sort the multiset of values $V = \{f(p): ~ p \in U \}$ to form a strictly increasing sequence $\mathsf{seq}
= (v'_1 <
 \ldots < v'_q)$, where, with each $i \in [q]$, we associate a weight $w_i$
 that equals the multiplicity of $v_i'$ in the multiset $V$ of
 values.
 \Comment{{Note that
 $\sum_{i \in [q]} w_i = u$.}}

\State We now partition the sequence $W = (w_1, \ldots ,w_q)$ into maximal
disjoint contiguous subsequences $W_1, \ldots W_{m''}$ such that for each $j \in [m'']$, either $\sum_{w \in
  W_j} w <  2u/m$, or $W_j$ contains only one member $w$ for
  which $w > u/m$. 
  
  \Comment{{This can be done greedily as follows. 
If $w_1 > u/m$ then $W_1$ will contain
only $w_1$, otherwise $W_1$ will contain the maximal subsequence
$(w_1, \ldots, w_i)$ whose
sum is at most $2u/m$. We then delete the members of
$W_1$ from $W$ and repeat the process. 
For $i \in [m'']$, let $w(W_i)$ denote the total weight in $W_i$.}}

Correspondingly, we obtain a
partition of the sequence $\seq$ of sampled values into at most $m''$ subsequences $\{\seq_j\}_{j \in [m'']}$. 
Some subsequences
contain only one value of weight at least $u/m$ and are called \emph{single-valued}.  
The remaining subsequences are called \emph{multi-valued}.

For a subsequence $\seq_j$, let $\alpha_j = \min(\seq_j)$ and $\beta_j = \max(\seq_j)$.
Let $\beta_0 = \inf(I)$.
Note that $\alpha_j \leq \beta_j$  and $\beta_{j-1} < \alpha_j$ for all $j \in [m'']$.

\State For $j \in [m'']$, we associate with the subsequence $\seq_j$, an interval $I_j \subseteq
\mathbb{R}$, where $I_j = (\beta_{j-1}, \beta_j] \cap I$, and an approximate density
$\widetilde{\den}(S,I_j) = w(W_j)/u$. The interval is multi-valued or single-valued depending on whether its corresponding sequence is multi-valued or single-valued, respectively.

\State For $j \in [m'']$, if the interval $I_j$ is the disjoint union of two contiguous intervals $I_j^{(1)}$ and $I_j^{(2)}$, then drop such an interval $I_j$ from consideration. 

\Comment{{ This situation can arise since $I$ is the disjoint union of several contiguous intervals and hence $I_j$ can contain points from two such consecutive and contiguous subintervals of $I$. In this case, by definition, $I_j$ is a multi-valued interval.}}

\State \textbf{Return} the set $\mathcal{I} = \bigcup_{\ell \in [m']} I_\ell$ of the remaining $m' \leq m''$ intervals.
\end{algorithmic}
\end{algorithm}

\begin{claim}\label{clm:layering}
If 
$\den(S,I) > 1/\log n$, then with probability 
 $1 - 1/n^{\Omega(\log n)}$, \text{Layering}$(S,I,m)$ returns a collection of intervals
 $\mathcal I = \{I_j\}_{j=1}^{m'}$ such that $\mathcal I$  is a nice $m$-partition of  $\bx(S,I)$. Furthermore, it 
makes a total of $m\log^4 n$ queries. 
\end{claim}
\begin{proof}

Since $\den(S,I) > 1/\log n$, a Chernoff bound implies that, with
probability at least $1 - \exp(-(m\log^3 n)/8)$, at least $m \log^2 n$ of the sampled points fall in
$\bx(S,I)$ and the algorithm does not fail in Step~\ref{stp:layering-fail}. In the rest of the analysis, we condition on this event happening.

To prove that $m' \leq 2m$, it is enough to bound $m''$, which is the total number of intervals formed before some multi-valued intervals are dropped at the last step.
The total number of intervals of weight at least $u/m$ is at most $m$ since the total weight is $u$. Other intervals have
weight less than $u/m$ and for each such interval $I_j$, it must
be the case that $I_{j-1}$ and $I_{j+1}$ are of weight at least $u/m$. It follows that $m' \leq m'' \leq 2m$.


We now prove that the family $\mathcal{I}$ output by Layering is a nice $m$-partition of $\bx(S,I)$.
It is clear from the description of~\cref{proc:layering} that the intervals output by the algorithm are disjoint.  
Let $\mathcal{B} = \{[a,b]: a,b \in I \text{ and } \exists v,w \in S \text{ such that } f(v) = a, f(w) =b\}$ denote the set of all true intervals of points 
from $\bx(S,I)$.  
Consider an interval $[a,b] \in \mathcal{B}$ such that $\den(S,[a,b]) \ge \frac{4}{m}$. The probability that 
less than $2u/m$ points from the sample have values in the range $[a,b]$ is at most $1/n^{\Omega(\log n)}$ by a Chernoff bound.
Conditioning on this event implies that for every $I_j, j\in [m']$ output as a multi-valued interval by the algorithm, we have $\den(S,I_j) < \frac{4}{m}$. 
Finally, for a single-valued interval $[a,a]  \in \mathcal{B}$ such that $\den(S, [a,a]) < \frac{1}{2m}$, with probability at least $1 - 1/n^{\Omega(\log n)}$, we have $\widetilde{\den}(S,[a,a]) \leq \frac{3}{2}\den(S, [a,a]) < \frac{3}{4m}$, where $\widetilde{\den}(S,I')$ denote the estimated density (as estimated in~\cref{proc:layering}) for a layer $\bx(S,I')$ when $I' \subseteq I$. 
Conditioning on this event implies that  for every $I_j, j\in [m']$ output as a single-valued interval by the algorithm, we have $\den(S,I_j) \geq \frac{1}{2m}$.

%


Finally, the number of layers that get dropped is at most $k$, each of them is multi-valued and hence, conditioning on the above events, the density of points lost in this process is at most $\frac{k}{2m} = o(1)$. Putting all of this together, we can see that the layers form a nice $m$-partition of $\bx(S,I)$.

The claim about the query complexity is clear from the description of the algorithm. 
\end{proof}

\subsection{Gridding}
Next, we describe the algorithm Gridding (see Algorithm~\ref{proc:preprocess}). 
\begin{algorithm}
\caption{Gridding$(S,I,m,\beta)$}
\label{proc:preprocess}
\begin{algorithmic}[1]
  \Require $S \subseteq [n]$ is a union of disjoint stripes,
  $I \subseteq R(f)$ is a disjoint union of intervals of values in
  $R(f)$, $D = \bx(S,I)$ is the domain on which we do gridding,
$m$ is a parameter defining the `coarse' grid size,
 $\beta < 1$ is a density threshold.

\hspace{-1.5cm}
\textbf{Output:} A grid of boxes $G_{m',m'},~ m' \leq 2m $ in
which there will be $\tilde{O}(m')$
 marked boxes.

\State Call Layering (\cref{proc:layering}) on 
inputs $S,I,m$. 
This returns, with high probability, a set $\mathcal{I}$ of
$m' \leq 2m$ value intervals $I = \bigcup_{j \in [m']} I_j$ that forms a nice $m$-partition  of $\bx(S,I)$.

\State Partition $S$ into $m'$ contiguous
intervals $S_1, \ldots S_{m'}$ each of size $|S|/m'$. This defines the
grid  of boxes  $D_{m', m'}=\{\bx(S_i, I_j):~ (i,j)\in [m']^2\}$ inside
the larger box $\bx(S,I)$.

\State Sample and query, independently at random, $\frac{\log^4 n}{\beta^2}$ points from each stripe
  $S_i, i \in [m'].$   For each $(i,j) \in [m']^2$, if $\bx(S_i, I_j)$
  contains a sampled point, then tag that box as \emph{marked}. 
If $\bx(S_i,I_j)$ contains at least $3\beta/4$ fraction of the sampled
  points in the stripe $S_i$, tag that box as {\em dense}.
  
  \State \textbf{Return} the grid $D_{m', m'}$ along with the tags on the various boxes.
\end{algorithmic}
\end{algorithm}

We note that initially, at the topmost recursion level of the algorithm
for $\pi$-freeness, we call Gridding with $S=[n],~ I = (-\infty,
+\infty)$ and our preferred $m$ which is typically $m = n^{\delta}$,
for some small $\delta <1$.

We prove in~\cref{cl:q_layering} that running
Gridding$(S,I,m)$ results in  a partition of $\bx(S,I)$ into a grid of
boxes $G_{m',m'}$ in which either the marked boxes 
contain a $\pi$-appearance, or the union of
 points in the marked boxes 
 contain {\em all} but an $\eta$ fraction of the points in
 $G_n$, for $\eta << \eps$.
 Additionally, with high probability, all boxes that are tagged {\em
  dense} have density at
least $\frac{1}{8}$-th of the threshold $\beta$ for marking a box as dense.

\begin{claim}\label{cl:q_layering}
Gridding$(S,I,m,\beta)$ returns a grid
of boxes  $D_{m',m'}$  that decomposes $\bx(S,I)$. It makes $\tilde{O}(m/\beta^2)$ queries, and with high probability,
  \begin{itemize}
    \item The set of intervals corresponding to the layers of $D_{m',m'}$ form a nice $m$-partition of $I$.
   \item  For every $i \in [m'],$
  either the stripe $\bx(S_i,I)$
        contains at least  $\frac{\log^2 n}{100\beta^2}$ marked boxes, or the number of
        points in the marked boxes in $\bx(S_i,I)$ is at least  $
        (1-\frac{1}{\log^2 n}) \cdot |S_i|.$
   \item Every box that is tagged dense has density at
     least $\beta/8$, and every box of density at least $\beta$ is
     tagged as dense.    
  \end{itemize}
  \end{claim}
\begin{proof}
  The bound on query complexity as well as the first item follows directly from~\cref{clm:layering}. The third item follows by a simple application of the Chernoff bound followed by a union bound over all stripes.

  For
  the second item, fix a stripe $S_i$ of $D_{m',m'}$. 
  Let $T \subseteq [m']$ be the set of all $j \in [m']$ such that $\bx(S_i,I_j)$ gets marked during Step 3 in $\rm{Gridding}$. 
 If $\sum_{j \in T}\den(S_i, I_j) \geq
 1-1/(\log^2 n)$ then we are done.
 Otherwise, each query  independently hits a box 
 that is not marked by any of the previous queries with probability greater than $1/(\log^2 n)$.
Thus, the expected number of boxes marked is at least $\log^2 n/\beta^2$.  Chernoff bound implies that, with probability at least $1 - n^{-\Omega(\log n)},$ at least $\frac{\log^2 n}{100\beta^2}$
 boxes are marked. 
 The union bound over all the stripes implies the second item.  
\end{proof}

\section{Generalized testing of forbidden
  patterns}\label{sec:generalizedtesting}
In this section, we formally define the problem of testing (or
deciding) freeness from $\nu$-appear\-ances with a certain leg-mapping. We then 
provide an algorithm for a relaxation of this 
testing problem.
Our algorithm for testing $\pi$-freeness is based on this. 
 A description of the algorithm, and
 a proof sketch for the case of patterns of length $3$ for specific leg-mappings is provided in~\cref{sec:example3}. 
 It illustrates some of the
   ideas for the general case, and it might be easier to follow. This
   is followed by an algorithm and a correctness proof for the most
   general case.

  Recall that $G_{n}$ denotes the $n \times |R(f)|$ grid that represents the input function $f:[n]
\to \R$. Let $G_{\ell,\ell}$ be a partition of $G_n$ into a grid of
 boxes for an arbitrary $\ell \geq 1$, and $D$ be a connected component in $G_{\ell,\ell}$
containing $t$ boxes $B_1, \ldots, B_t$.
Let $\nu \in \mathcal S_r$, and let $\phi: [r] \mapsto \{B_1,\dots,B_t\}$ be
an arbitrary mapping of the legs of $\nu$ into the boxes of
$D$, where $t \leq r$. We say that $1 \leq i_1 < \ldots < i_r \leq n$ is a
$\phi$-legged $\nu$-appearance if $(i_1, \ldots ,i_r)$ forms a
$\nu$-appearance in $G_n$ such that the point
$(i_j,f(i_j))$ is contained in the box $\phi(j)$ for each $j \in [r]$.  That is, the
legs of the $\nu$-appearance are mapped into the boxes
given by $\phi$. For example, consider Figure~\ref{fig:1}(B),  $\nu = (3,2,1,4),$ and $D$  the 
component formed by the two boxes in the same layer.
The function $\phi$ maps the $3$-leg and $2$-leg of the $\nu$-appearance to the
left box and the $1$-leg and $4$-leg to the right box. 
The connected component $D$ is $\phi$-legged $\nu$-free if it contains no $\phi$-legged $\nu$-appearances.
It is $\eps$-far from being $\phi$-legged $\nu$-free if the values of at least $\eps \cdot |\bigcup_{j \in [t]} \mathrm{St}(B_j)|$ points belonging to $D$ must be modified in order to make $D$ free of $\phi$-legged $\nu$-appearances, where $\mathrm{St}(B)$ for a box $B$ denotes the stripe corresponding to $B$. Note that a function could be $\phi$-legged $\nu$-free but very far from being $\nu$-free. For example, for the $\phi$ referred to above in~\cref{fig:1}(B), it could be that there are many appearances of $(3,2,1,4)$ which are all in the left box or all in the right box or both, but there are no appearances with the leg mapping $\phi$.

 The property of being free of
 $\phi$-legged $\nu$-appearances is a generalization of the property of
 $\pi$-freeness. Taking $\ell=1,$ $G_{\ell,\ell}$ is just $G_n$
 itself viewed as one single box $D$. When $\nu = \pi$ and $\phi$ is the constant function that maps each leg to the unique box
$D$, any $\pi$-appearance in $G_n$
 is a $\phi$-legged $\nu$-appearance.
 
The problem of testing $\phi$-legged $\nu$-freeness was not previously explicitly studied and we believe that it is an interesting research direction in its own right. 
 Even though its complexity is not known, we encounter it only as a subproblem in the testing of
 standard $\pi$-freeness. This motivates the following definition.
 
\begin{definition}\label{def:generalized-testing}
 Let  $\pi \in \mathcal S_k$, $\nu \in \mathcal{S}_r$, where $r \leq k$. Let $G_n$ denote the $n \times |R(f)|$ grid that represents the input function $f:[n] \to \R$. 
 For $\ell \geq 1$, let $G_{\ell,\ell}$ be a
 decomposition of $G_n$ into boxes. For $t \leq r$, let $D$ be  a $t$-boxed
 single component composed of the boxes  $B_1, \ldots
   ,B_t$ in 
   $G_{\ell,\ell}$ and  let $\phi : [r] \mapsto \{B_1, \ldots ,B_t\}$. The problem 
 $\mathsf{GeneralizedTesting}_{\pi}(\nu,\phi,D)$ is the following. For a parameter $\eps \in (0,1)$, if $D$ is $\eps$-far from being $\phi$-legged $\nu$-free, 
  find a $\phi$-legged
 $\nu$-appearance in $D$ \textbf{OR} find any (unrestricted) $\pi$-appearance . 
 \end{definition}

 Our algorithm for $\mathsf{GeneralizedTesting}_{\pi}(\nu,\phi, D)$ is called $\mathsf{AlgTest}_\pi(\nu,\phi,D,m,\eps)$ and is presented in~\cref{algo:main}. The algorithm has a permutation $\pi \in \mathcal{S}_k$ hardwired into it. It gets oracle access to a function $f:[n] \to \R$ and
its inputs are (1) $\nu \in \mathcal{S}_r$, $r \leq k$, (2) a component~$D$ composed of the boxes $B_1, \dots, B_t, ~ t
 \leq r, $ in a grid
 $G_{\ell,\ell}$ of $G_n$, where $\ell \geq 1$, (3) a mapping
 $\phi:[r] \to \{B_1, \dots B_t\}$, (4) a distance parameter
 $\eps \in (0,1)$, and (5) a free parameter $m \geq 2$. 
 The parameter $m$ is used  to control the 
 query complexity. 
We are not specifying $\ell$ explicitly here, but it is implicit in the way boxes of $D$ are defined.
  If $D$ is $\eps$-far from being free of
 $\phi$-legged $\nu$-appearances, with high probability, the algorithm either finds a
 $\pi$-appearance or a $\phi$-legged $\nu$-appearance in $D$.

   \begin{algorithm}
\caption{AlgTest$_{\pi}(\nu,\phi,D, m,\eps)$} 

\begin{algorithmic}[1]

\Require pattern $\nu \in \cS_r$; 
$D$ is a component containing 
boxes $B_1, \ldots, B_t$ in $G_{\ell, \ell}$ for $t \in [r]$;
the function $\phi: [r] \mapsto \{B_1, \ldots ,B_t\}$ is a leg-mapping of $\nu$ into the boxes of $D$; parameter $m \leq n$; parameter $\eps \in (0,1)$.

\noindent
\hspace{-0.9cm}
\textbf{Goal:} Find a $\phi$-legged $\nu$-appearance \textbf{or} an
unrestricted $\pi$-appearance
in $D$.


\State Let $S = \bigcup_{i \in [t]} \St(B_{i})$ and $I = \bigcup_{i
  \in [t]} L(B_{i})$ define $\bx(S,I)$ in $G_{\ell,\ell}$ that contains
$D$. \label{stp:box-def} 
\vspace{-0.7cm}
\State\textbf{Base cases:} \textbf{Call} $\mathsf{BaseCaseAlgTest}_{\pi}(\nu,\phi,D,S, m,\eps)$ and output what it outputs. 
\label{base-case}
\State \textbf{Gridding  $D$:}  We set $\beta = \frac{\eps}{200k\kappa(k)}$. Call Gridding$(S,I,m,\beta)$ which returns  a decomposition of $\bx(S,I)$ 
into an $m' \times m' $ grid $D_{m',m'}$  of
subboxes, where $m \leq m' \leq 2m$. 
A subset of these boxes in $D_{m',m'}$ are \emph{marked} and a subset of the marked boxes are \emph{dense}. \label{stp:gridding}

\State \textbf{Simple case:} If $D_{m',m'}$ contains more than $\kappa (k) \cdot m'$
marked subboxes then \textbf{output} ``$\pi$-appearance is found''. \label{stp:too-many-boxes}

\State \textbf{Sparsification:} Delete each stripe and layer in $D_{m',m'}$
that contains more than $d=100k \kappa(k)/\eps$ marked
subboxes. Delete all non-dense subboxes. \label{stp:sparsification}

\State \textbf{Multi-component configurations:} 
Let $c = r^{3r}$ denote an upper bound on the number of distinct configurations with at most $r$ components.
For each $(\phi,\nu,D)$-consistent configuration $\cC$ (see~\cref{def:phi-consistence}) with $p > 1$ many components $\cC_1, \ldots \cC_p$, sub-permutations $\nu_1, \ldots ,\nu_p$
of $\nu$ and mappings $\phi_1, \ldots , \phi_p$: \label{stp:multi-comp}

\begin{enumerate}
\item Recursively \textbf{call} $\mathsf{AlgTest}_{\pi}(\nu_i,\phi_i,D_i,m,\eps')$ with distance parameter $\eps' = \frac{9\eps}{10kcr^2 \cdot r! \cdot (2d)^r}$ for {\em every} component $D_i$, where $D_i$
is a copy of $\cC_i$  in $D_{m',m'}$, and is contained in $D$. Note that $\nu_i$'s are smaller patterns.

\item \textbf{Output}  ``$\phi$-legged $\nu$-appearance is found'' if for  a copy  $(D_1, \ldots ,D_p)$ of
$(\cC_1, \ldots ,\cC_p)$, for each $i \in [p]$, $D_i$ contains a $\phi_i$-legged $\nu_i$-appearance. 
If a $\pi$-appearance is found
among the sampled points, \textbf{output} ``$\pi$-appearance is found'' .
\end{enumerate}
  
\State \textbf{Single  component configurations:}  
Let $\mathcal A$ be the set of all possible copies in $D$ of $(\phi,\nu,D)$-consistent single-component configurations $\cC$ in
$D_{m',m'}$. \label{stp:single-comp} 
 
\begin{enumerate}
\Loop~$\frac{\log^3 n}{\eps^{r}}$ times:
\State \textbf{Sample} a member $D'$ 
from $\mathcal A$ uniformly at random, and 
for each $\phi$-consistent mapping $\phi'$ (\cref{def:phi-consistence}), \textbf{call} $\mathsf{AlgTest}_{\pi}(\nu,\phi',D',m,\eps'')$
with $\eps'' = \frac{9\eps}{20k \cdot (2d)^r \cdot (r-1)! \cdot r^r}$. 

\EndLoop
\end{enumerate}

\State \hspace{-0.6cm} If no output is declared in any of the previous steps, \textbf{output} ``not found''.
\end{algorithmic}
\label{algo:main}
\end{algorithm}

 \begin{algorithm}
\caption{BaseCaseAlgTest$_{\pi}(\nu,\phi,D,S, m,\eps)$} 
\begin{algorithmic}[1]

\Require pattern $\nu \in \cS_r$; 
$D$ is a component containing 
boxes $B_1, \ldots, B_t$ in $G_{\ell, \ell}$ for $t \in [r]$; $S$ is a set of indices encompassing the points in $D$;
the function $\phi: [r] \mapsto \{B_1, \ldots ,B_t\}$ is a leg-mapping of $\nu$ into the boxes of $D$; parameter $m$; parameter $\eps \in (0,1)$.

\noindent
\hspace{-0.9cm}
\textbf{Goal:} Find a $\phi$-legged $\nu$-appearance \textbf{or} an
unrestricted $\pi$-appearance
in $D$.

\State  If $|S| \leq  m$ query all
indices in  $S$. \textbf{Output}  ``$\pi$-appearance is found''  or 
``$\phi$-legged  $\nu$-appearance is found'' if one of these is found.
 \State If $k=1$, \textbf{output} ``$\pi$-appearance is found'' if $D$ contains
  a point.

\State If $r = 2$, use the test for restricted appearance of
  $2$-patterns as described in the proof of~\cref{lm:2-pattern}. 
  If $r=1$,
  \textbf{output} ``$\phi$-legged  $\nu$-appearance is found''  if the box $\phi(1)$
  contains a point.

\State If the sampled points in $D$
  already contain a $\phi$-legged $\nu$-appearance, or contain a $\pi$-appearance then \textbf{output} ``$\phi$-legged  $\nu$-appearance is found'' or ``$\pi$-appearance is found'' respectively.

\end{algorithmic}
\label{algo:base-case}
\end{algorithm}

 The algorithm is recursive.  
 A recursion is done by reducing $\nu$ to smaller length patterns, and/or
 self-reduction to the same $\nu$ but on a smaller size box
 $D'$. The important base cases  (see \cref{base-case}) are when the size of $D$ is small enough to allow
 queries to all indices in $D$, or when  $\nu \in \mathcal S_2$, in
 which case the algorithm is reduced to testing monotonicity.

 We recall that the permutation $\pi \in \cS_k$ is fixed and hardwired into the
 algorithm. The grid $G_n$ is fixed and not part of the
 recursion.  The algorithm makes its queries to the $t$-boxed component $D$ in a grid
 of boxes $G_{\ell,\ell}$ defined with respect to $G_n$ (that is, a subfunction of the original function $f$).  
 The first main step of the algorithm is to grid the region $\bx(S,I)$ with parameter $m$ into a grid $D_{m',m'}$ of subboxes, where $m' = O(m)$ and
 $S \subseteq [n]$ and $I \subseteq R(f)$ are the unions of the sets of all indices
  and values, respectively, in the $t$ boxes of $D$. 
In this process of refining the existing boxes of $D$ into subboxes, the legs of a $\phi$-legged $\nu$-appearance
 in $D$ are mapped into subboxes formed by $D_{m',m'}$. 
 The set of subboxes that contain the legs of a particular $\phi$-legged $\nu$-appearance
 can form a different configuration (with one or more connected components in it) than the 
 configuration corresponding to the component $D$. This prompts 
 us to make the following definitions which are used in the algorithm description.
 
 \begin{definition}\label{def:phi-consistence}
 Consider a configuration $\mathcal{C}$ consisting of components $C_1, \dots C_p$ for $p \in [r]$.
 Additionally, for $i \in [p]$ let $\phi_i$ be a mapping from the set of legs of a permutation $\nu_i$ to the set of boxes in $C_i$.
 The configuration $\mathcal{C}$ along with $\{\nu_i\}_{i \in [p]}$ and $\{\phi_i\}_{i \in [p]}$ is $(\phi, \nu, D)$-consistent
 if (1) the subboxes of the grid $D_{m',m'}$ contain a copy of $\mathcal{C}$ and (2) a union of the legs of $\phi_i$-legged $\nu_i$-appearances in the copies of 
 $C_i$ form a $\phi$-legged $\nu$-appearance in the corresponding copy of $\mathcal{C}$. 
 
 If $p = 1$, then $\nu_1 = \nu$ and we say that the mapping $\phi_1$ is simply $(\phi,D)$-consistent.
 \end{definition}
  
 In order to exemplify these definitions, let $\nu = (1,3,2)$ and let $D$ consist of two boxes in the same layer and let $\phi$ map the $1,3$ legs to the box on the left and the $2$ leg to the box on the right. The $(1,3,2)$-appearances in the green boxes in~\cref{fig:2} illustrate this. These appearances can belong to various possible configurations upon further gridding of the two boxes as illustrated by the various cases shown in the same figure. Specifically, the smaller orange boxes are representative of subboxes obtained upon gridding of the two green boxes.~\cref{fig:2}(A) shows a $(\phi, \nu, D)$-consistent configuration composed of three components and~\cref{fig:2}(B)-(D) show $(\phi,\nu,D)$-consistent configurations composed of two components.~\cref{fig:2}(E)-(H) show $(\phi, \nu, D)$-consistent configurations with just a single component and in these cases, the leg mappings are $(\phi,D)$-consistent.

We further note that our
 algorithm  does not use
 any  structure of $\pi$. The only role of
 $\pi$ in the algorithm is to ensure that after gridding, the resulting 
 grid $D_{m',m'}$ contains only $O(m)$ marked boxes as
 otherwise, by~\cref{clm:MarcusTardos}, a $\pi$-appearance is
 guaranteed.

The following theorem asserts the correctness of $\mathsf{AlgTest}_{\pi}(\pi, \phi,G_n,m,\eps)$ and the
corresponding query complexity.

\begin{theorem}\label{thm:correctness-main}
  Let $\pi \in \mathcal{S}_k$ and $\nu \in \mathcal{S}_r, ~ r \leq k.$  
  Let $f:[n] \to \R$ and let $G_n$ denote the $n \times |R(f)|$ grid of function points. 
  Let $\ell \geq 1$ and $G_{\ell, \ell}$ be an $\ell \times \ell$ grid decomposing $G_n$.
  Let $D$ be a connected component in $G_{\ell,\ell},$ composed of
  boxes $B_1, \dots ,B_t, ~t \leq r$ and  $\phi: [r] \to \{B_1, \dots B_t\}$.
  Let $S = \bigcup_{i \in [t]} \St(B_{i})$ and $I = \bigcup_{i
  \in [t]} L(B_{i})$.
  Let $\eps \in (0,1)$. Let $m = k|S|^{\eta}$
   for $\eta \in \left(\Omega\left(\frac{1}{\log \log \log n}\right),1\right)$. 
   Let $a$ be the smallest integer such that $m^a \geq k^{a-1}|S|$.
  If $D$ is $\eps$-far from $\phi$-legged
  $\nu$-freeness, then $\mathsf{AlgTest}_{\pi}(\nu,\phi,D,m,\eps)$ finds either a $\phi$-legged $\nu$-appearance or a $\pi$-appearance, with
  probability at least $1 - o(1)$. Its query complexity is 
  $\tilde{O}\left(m^r \left(\frac{k}{\eps}\right)^{\Theta(k^a)}\right)$, where the $\tilde{O}(\cdot)$ notation hides polylogarithmic factors in $n$.
\end{theorem}

We note that since $\mathsf{AlgTest}_\pi(\nu,\phi,D,m,\eps)$ either finds a $\pi$-appearance
or a $\phi$-legged $\nu$-appearance in $D$, then if $D$ is free of
$\phi$-legged $\nu$-appearances, the algorithm will never return
such an appearance. 

Our $\pi$-freeness tester is simply $\mathsf{AlgTest}_\pi(\pi,\phi,G_n,m,\eps)$, where $\phi$ is the constant function mapping each leg 
to the entire grid $G_n$ and $m = kn^{1/a}$ for an integer parameter $a \leq \log \log \log n$ that we can control.

\begin{corollary}
There is a $1$-sided error test for $\pi$-freeness of functions of the form $f:[n] \to \R$, for every $\pi \in
\mathcal S_k$, with query-complexity $\tilde{O}(\left(\frac{k}{\eps}\right)^{\Theta(k^a)}\cdot n^{k/a})$, for integer $a \leq \log \log \log n$.
  \end{corollary}

\subsection{Proof of Correctness}\label{sec:proof-correct1}
In~\cref{sec:example3}, we start with  a description of the algorithm and the proof sketch for the first
non-base case of testing $\phi$-legged $\nu$-freeness for $\nu \in \mathcal{S}_r$, $r=3$, with respect to an arbitrary $\pi \in \mathcal S_k$ and fixed
$k \geq 4$. 
In~\cref{sec:full_proof}, we present
the proof  of~\cref{thm:correctness-main}. 

\subsubsection{An example for \texorpdfstring{$\nu \in \mathcal S_3$}{nu in S3}}\label{sec:example3}
For this exposition, we  fix $\nu = (1,3,2)$, and $D$ being
composed of $2$ boxes $B_1$, $B_2$ in the same layer, 
where
$B_1$ is to the left of $B_2$, and $\phi$ maps the $1,3$ legs of
$\nu$ to $B_1$, and the $2$-leg to $B_2$.  See~\cref{fig:2}(D) for an
illustration of one such case. In the figure, the 
green boxes represent $B_1$ and $B_2$. 
The orange boxes indicate the subboxes in the finer 
grid formed when gridding is called on the green boxes.

\begin{figure}
\begin{minipage}{0.48\textwidth}
\begin{center}
\includegraphics[scale = 0.55]{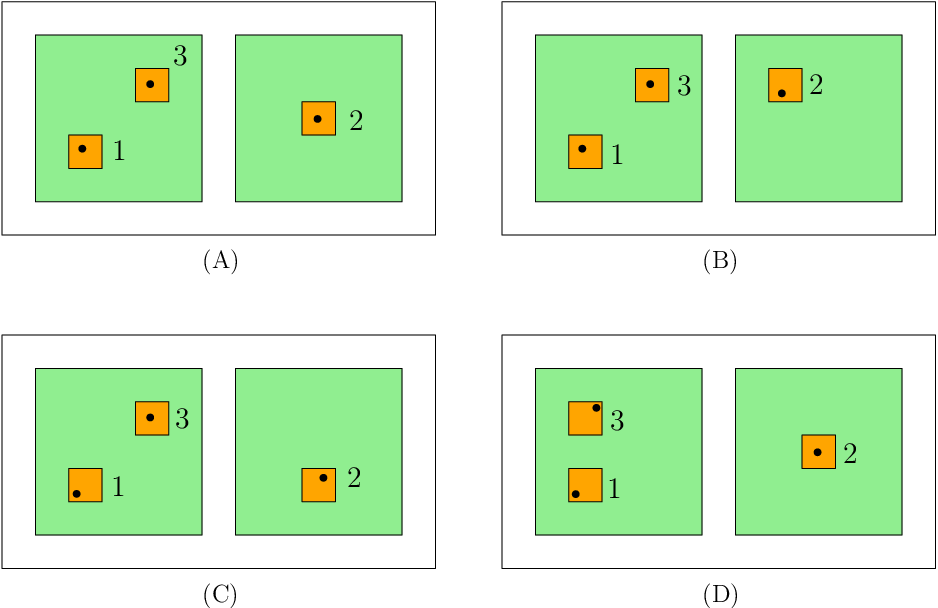}
\end{center}
\end{minipage}
\vline
\vline
\begin{minipage}{0.48\textwidth}
\begin{center}
\includegraphics[scale = 0.55]{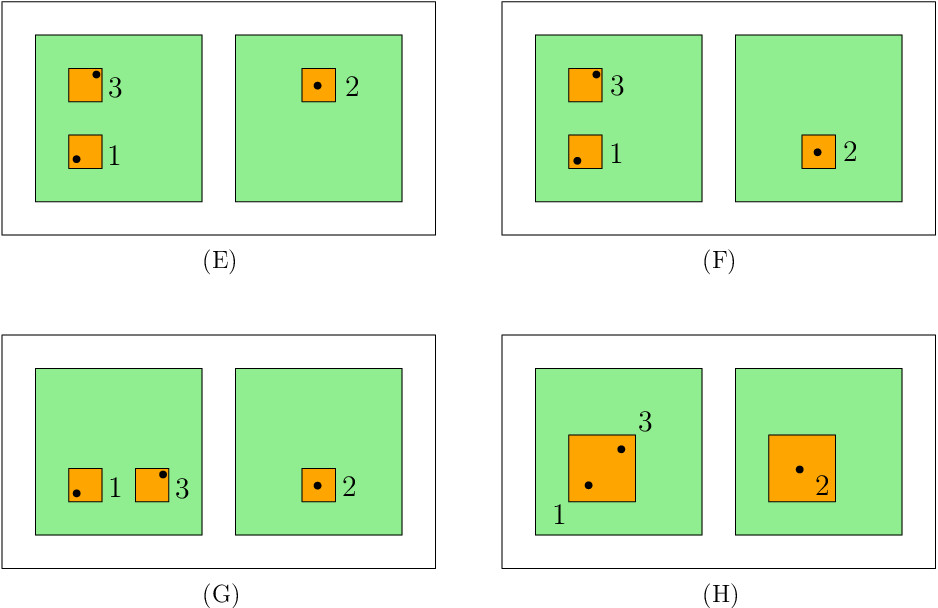}
\end{center}
\end{minipage}
\caption{$(1,3,2)$-appearances with legs spread across two green boxes sharing a layer resulting in new configurations upon further gridding of the boxes into smaller orange boxes.}
\label{fig:2}
\end{figure}

We note that \cref{fig:2}(D) illustrates the hardest case for $\nu \in \mathcal S_3$. There
are additional one-component configurations in which the boxes are in the same
stripe or layer, but these turn out to be  much easier.
We will set $m = m(n)$ to be defined later
and express the complexity as a function of $m$. 
We do not specify $\pi$ since, as explained above, $\pi$ is only
needed at Step~\ref{stp:too-many-boxes} of the algorithm when
the number of marked boxes is superlinear in $m$ in some recursive call.
The argument here holds for any $\pi \in \mathcal
S_k, ~ k\geq 4$.

\paragraph{Algorithm to test $\phi$-legged $\nu$-freeness.} Let $\nu=(1,3,2)$ and $\phi$ be such that $\phi(1) = \phi(3) = B_1$ and $\phi(2) = B_2$.
\begin{enumerate}
\item We assume that $B_1,B_2$ are over $s \leq n$ indices each, and
  that the distance of $B_1 \cup B_2$ from
  $\phi$-legged $\nu$-freeness is at least $\eps =
  \Omega(1)$. In particular $B_1, B_2$ are dense. 
In Step~\ref{stp:gridding} of~\cref{algo:main}, we
grid the appropriate box containing $B_1 \cup B_2$ (as defined in~Step~\ref{stp:box-def} of~\cref{algo:main})  into a  $m' \times
  m'$ grid, 
  $D_{m',m'}$, of  subboxes (each over $2s/m'$ indices), where $m \leq m' \leq 2m $.
 We either find a $\pi$-appearance among the sampled points or we may assume, after~Steps~\ref{stp:too-many-boxes} and \ref{stp:sparsification} 
that there are
  $O(m')$ dense subboxes in $D_{m',m'}$ and that each layer and
  each stripe contains $O(1)$ dense boxes. The latter claim is obtained by an
  averaging argument and is described in the formal proof
  in~\cref{sec:full_proof}. The argument is that if $B_1 \cup
  B_2$ contains a large matching of $\phi$-legged $\nu$-appearances, then so
  does the restricted domain after deleting points from non-dense boxes as well as
  and deleting layers and stripes that contain too many dense boxes from $D_{m',m'}$. These
  steps take $\tilde{O}(m)$ queries overall, which is the complexity of the algorithm Gridding.

  \item A $\phi$-legged $\nu$-appearance in $B_1 \cup B_2$ can be in $8$
    possible configurations in the grid $D_{m',m'}$, as depicted in~\Cref{fig:2}.
Consider first $\cC_1, \ldots ,\cC_4$ as in \Cref{fig:2}(A)-(D), that form $2$ or $3$ components each. For these,
    a $\phi$-legged $\nu$-appearance in $B_1 \cup B_2$ decomposes into two or three
    subpatterns, and for which any restricted appearances in the
    corresponding components results in a $\phi$-legged
    $\nu$-appearance. For example,  in \Cref{fig:2}(B) the configuration $\cC_2$ contains one component $D_1 =
    (B_{1,3}, B_{2,2})$, where $B_{1,3} \in B_1, B_{2,2} \in B_2$, and
    another single boxed component $B_{1,1} \in B_1$, where $B_{i,j}$ is the 
    orange subbox contained within the green box $B_i$ and such that the $j$-th leg
    belongs to $B_{i,j}$ for $i\in [2], j \in [3]$.

    In Step~\ref{stp:multi-comp} of~\cref{algo:main}, we
    test each of the $O(m)$ many copies of $D_1$ for a $\phi'$-legged $(2,1)$-appearance for
    which $\phi'(2)= B_{1,3}$ and $\phi'(1)= B_{2,2}$. Then for
    any such $D_1$-copy in which such a $\phi'$-legged $(2,1)$-appearance is
    found, any nonempty dense box $B_{1,1}$ forming with $D_1$ a copy
    of $\cC_2$ results in a $\phi$-legged $\nu$-appearance.

    Since this is a reduction to generalized $2$-pattern appearance, the recursion
    stops here with $O(\log n)$-complexity per copy of $D_1$. Hence,
    altogether this will contribute a total of $\tilde{O}(m)$ queries.
    Procedures along the same lines work for any of $\cC_i, ~ i=1,2,3,4$.

   If a desired $\phi$-legged $\nu$-appearance (or a $\pi$-appearance) is found in the above process, then clearly a correct output is produced.

    On the other hand, if indeed $(B_1,B_2)$ contains $\Omega(s)$ (that is,
    linear in the size of $B_1 \cup B_2$) many $\phi$-legged
    $\nu$-appearances that are consistent with one of the configurations
    $\cC_i, i \in [4]$, then,  by an averaging argument, there will be such a $D_1$ and
    corresponding $B_{1,1}$ that together contribute $\Omega (s/m)$
    (that is, linear in the domain size of  $D_1$) such subpattern
    appearances.

    We note that for the more general case of $r > 3$,
    the reduction will be done in higher complexity 
    per component (that is dependent on $m$ rather than just $O(\log n)$). 

  \item Consider now a consistent configuration $\cC_i$ for $i=5,6,7,8$ that forms a single component
    (with $2$ or $3$ orange subboxes) as illustrated in~\cref{fig:2}(E)-(H).  In these cases, if such appearances
    contribute $\eps'$ to the total distance, then a simple averaging
    argument shows that for a uniformly sampled component, its
    distance from $\phi$-legged
    $\nu$-freeness will be linear. Hence in Step \ref{stp:single-comp}, sampling
    such a component will enable us to recursively find a $\phi$-legged
    $\nu$-appearance with high probability. Since the size of a
    component on which the recursive call is made is $\Theta(s/m)$,
    the complexity of this step is $\tilde{O}(q(s/m, \eps'))$, where
    $q(s,\delta)$ is the complexity of the algorithm, for the case of $\nu
    \in \mathcal S_3$, in terms of the size $s$ of $D$, and a
    distance parameter $\delta$.
\end{enumerate}

\paragraph{Correctness.} The correctness of the algorithm follows from the fact that if $D$ is
indeed far from being $\phi$-legged $\nu$-free, then it must be that
there are linearly many $\phi$-legged $\nu$-appearances
in at least one of the $8$ configurations discussed
above, and for each case, either a $\pi$-appearance or a
$\phi$-legged $\nu$-appearance is found, by induction. Note however,
that there is a drop in the distance parameter from $\eps$ to $\eps'$, due to the deletion of 
points in Step~\ref{stp:sparsification} of the algorithm, and the averaging arguments
resulting in the call with smaller distance parameters at Step~\ref{stp:multi-comp} and Step~\ref{stp:single-comp}. 
This does not matter as long as $\eps'$ is kept constant (or even $
1/\log n$), forcing the recursion depth to be bounded from above by a constant.

\paragraph{Query complexity.}We now analyze the query complexity of the algorithm for the special case described above.
The parameter $m$ is to be interpreted as the \emph{query budget} of the algorithm.
We abuse notation and use $s$ to indicate the total number of indices that the component $D$ contains.
Let $a$ be the smallest integer such that $m^a \geq s$.
This parameter $a$ denotes the recursion depth of our algorithm and 
we express our recurrence relation in terms of $a$.
Let $t(m,a)$ denote the query complexity of 
the above algorithm with parameter $m$ for functions over a domain of size $s \leq m^a$.  
We omit the dependence of the query complexity on $\eps$ and assume that $\eps = \Theta(1)$ for the purposes of this high level description.

For the base case, we have $a = 1$. Then, $q(m,1) = m = \Theta(s)$ since the algorithm can query all the indices and still be within the query budget.
\FloatBarrier
If $a > 1$, ignoring polylog factors, we have $t(m,a) = m + m + t(m,a-1)$. The
first summand here is the number of queries made by the Gridding. The
second summand is the number of queries made by Step 2 above (corresponding to Step~\ref{stp:multi-comp} in \cref{algo:main}). The last summand denotes
the query complexity of the recursive call on a subbox of size $\Theta(s/m)$
with the same $m$, for which the recursion depth is $a-1$.


The recurrence implies that $t(a,m)=\tilde{O}(am)$, which
implies a query complexity $s^{\delta}$ by choosing $m=
s^{\delta}$. We note that for $\delta = \Omega(1)$  the recursion
depth is $a = 1/\delta = O(1)$ as indeed needed to keep the distance
parameter constant. Moreover, for $\delta =
1/\log\log\log n$, the distance parameter $\eps' =
\Omega(1/\log n)$ at all recursion levels and the complexity becomes $s^{o(1)}$.
 \subsubsection{Formal Proof of Theorem \ref{thm:correctness-main} for
 general \texorpdfstring{$\nu \in \mathcal{S}_r$}{nu in Sr}}\label{sec:full_proof}

We now provide the formal proof of Theorem \ref{thm:correctness-main}. In
 what follows, we refer to the steps in the description of
 $\mathsf{AlgTest}_\pi(\nu,\phi,D,m,\eps)$ (see~\cref{algo:main}).
 Let $S, I$ be as defined in the algorithm (based on the component that
 $D$ contains). The proof is by induction on the parameters $a$ and~$r$ as defined in the 
 statement of Theorem \ref{thm:correctness-main}. To recall,
 $a$ is  the
 smallest integer for which $m^{a} \geq k^{a-1} |S|$ and indicates the recursion depth of the algorithm.

\paragraph{Base Cases.} For completeness we start with the base cases (see Algorithm~\ref{algo:base-case}). One of them is when $a = 1$, which is equivalent to
$m \geq |S|$, in which case the algorithm queries all indices in $S$ and solves the problem correctly with probability $1$.
The other base case is when $r = 2$, which is the same as
 testing for restricted $\nu$-appearance
for a $2$-pattern $\nu$. 
That is, the input is $\nu
\in \mathcal{S}_2$, a component~$D$ with at most two boxes, a leg
mapping function $\phi$ and the
distance parameter $\eps$.  There is no need of $\pi$ as we will show how to test
$\phi$-legged $\nu$-freeness
unconditionally. 
Lastly, the only $2$-pattern up to isomorphism is $\nu=(2,1)$ which is
assumed to be the input.

\begin{lemma}
  \label{lm:2-pattern}
  Let $D$ be a connected component in $G_n$ and $\phi$ a leg
  mapping for $\nu=(2,1)$ into the boxes of $D$. Let $S \subseteq [n]$ denote the set of indices belonging to the boxes in $D$. For any $\eps > 0$
  there is a $1$-sided error $\eps$-tester for $\phi$-legged
  $\nu$-freeness in $D$ with query complexity $O((\log |S|)/\eps)$.
  \end{lemma}
  \begin{proof}
  The component $D$ has at most $2$ boxes since $\nu$ is a pattern of length $2$.
  If $D$ is a single box $\bx(S,I)$, then the problem is identical to erasure-resilient monotonicity testing,
  where the points belonging to $D$ are the nonerased points and the points $(x,y)$ with $x \in S$ and $y \notin I$ are erased. 
  In this case, we can use an existing $O((\log |S|)/\eps)$-query erasure-resilient tester~\cite{DixitRTV18}, since $D$ is dense and contains a constant fraction of 
  points in the stripe defined by $S$.

    In the rest, we assume that $D$ is composed of exactly two boxes $D = B_1 \cup
    B_2$, and $\phi(i)= B_i, ~i=1,2$. 
    Consider first the case where the boxes are on the same layer and $B_1$ is on the left of $B_2$. A $\phi$-legged $\nu$-appearance in this case is constituted by $(i,j) \in \St(B_1) \times \St(B_2)$ such that
    $f(i) > f(j)$. The $\eps$-tester is as follows. 
    \begin{enumerate}
    \item Sample $
    O(1/\eps)$ indices independently and uniformly at random from
    each one of the stripes $\St(B_1)$ and $\St(B_2)$. 
    \item Reject if there exists indices $i,j$ in the sample such that $(i,f(i)) \in B_1$, $(j,f(j)) \in B_2$ and $f(i) > f(j)$; accept otherwise.
    \end{enumerate}

    The tester has $1$-sided error and has query complexity
    $O(1/\eps)$. 
    We now show that if $D$ is $\eps$-far from $\phi$-legged $\nu$-freeness, 
    then the tester above rejects with constant probability.
    Let $s = |\St(B_1)|$. It must be the case that $D$ has a 
    matching $M$ of $\phi$-legged $\nu$-appearances of cardinality at least $\eps s$.
    Let $\alpha$ be the median value of the $2$-legs in this
    matching. Namely, there are at least $\eps s/2$ pairs in
    $M$ with  the value of the left leg $> \alpha$. Thus, the probability of a sampled index $x
    \in \St(B_1)$ to be a $2$-leg in $M$ and with $f(x) > \alpha$ is at
    least $\eps/2$.  By the same argument, for half the pairs in
    $M$ their $2$-leg value is below $\alpha$ and, for each such
    pair, its corresponding $1$-leg in $B_2$ has a lower value than
    $\alpha$. Hence with probability at least $\eps /2$, a random query $y
    \in B_2$ will be such that $f(y) < \alpha$.
    We conclude that if these two events occur we find the required
    pair. These two
    events happen with probability at least $1- 2(1-\eps/2)^\ell
    > 2/3$ for an appropriate number of queries $\ell = O(1/\eps)$. This ends the proof
    for this case.

The other case is when $B_1$ and $B_2$ are on the same stripe. A
similar tester with a similar correctness argument is applicable for this case as well.
  \end{proof}

 \paragraph{General Case.}

 Let $\pi \in \mathcal{S}_k$ be fixed and let $\nu \in \mathcal{S}_r$, where $r \leq k$. 
 Assume that we call the algorithm
 $\mathsf{AlgTest}_{\pi}(\nu,\phi,D,m,\eps)$, where $D$ is a single component (in some grid of boxes
 $G_{\ell,\ell}$)   containing 
 the boxes $B_1, B_2, \dots, B_t, ~ t \leq r$.
 Let $S = \bigcup_{j \in [t]} \St(B_j)$ and $I = \bigcup_{j \in [t]} L(B_j)$.

 In what follows, we show that if $D$ is $\eps$-far from
 being $\phi$-legged $\nu$-free, the call to the algorithm $\mathsf{AlgTest}_\pi(\nu,\phi,D,m,\eps)$ finds a
 $\phi$-legged $\nu$-appearance or a $\pi$-appearance w.h.p. 
 This will complete the proof of correctness.
 We assume, for simplicity, that $f$ is one-to-one (see note at the end
 of this section for handling the case when $f$ is not one-to-one).


The first
 step of~\cref{algo:main} is Step~\ref{stp:gridding}, which is a call
 to Gridding$(S,I,m,\beta)$, where $\beta= \eps/(200k\kappa)$. By~\cref{cl:q_layering}, we know that w.h.p.\ this call returns a
 decomposition of $\bx(S,I)$ into an $m' \times m' $ grid of
 subboxes $D_{m',m'}, ~m \leq m' \leq 2m$, where a subset of boxes are
 marked and a subset of these marked boxes are dense w.r.t.\ the threshold
 $\beta$. 
 Additionally, the set of intervals $\mathcal{I} = \{I_j\}_{j \in [m']} $ corresponding to the layers of $D_{m',m'}$ form a nice 
 $m$-partition (see~\cref{def:nice-partition}) of $\bx(S,I)$.  Since $f$ is one-to-one, there are no 
 single-valued layers and hence, for each $j \in [m']$, it holds that $\den(S,I_j) \leq 4/m$.  

 In the next stage (Step~\ref{stp:too-many-boxes} of
 $\mathsf{AlgTest}$), the algorithm checks whether the marked boxes of
 $D_{m',m'}$ directly contain a $\pi$-appearance. Such an appearance 
 corresponds to an actual appearance in $f$ by~\cref{obs:homo}. Hence, we
 either find a $\pi$-appearance and we are done, or we conclude by~\cref{clm:MarcusTardos} that $D_{m',m'}$ has at most $\kappa m'$
 marked boxes.  
 Then, in Step~\ref{stp:sparsification} of $\mathsf{AlgTest}$ we delete all points in each layer and
 each stripe that  
 contains more than $d$ marked boxes. We additionally delete all points in all the non-dense boxes.

\begin{claim}
If $D$ is $\eps$-far from $\phi$-legged $\nu$-free, then the union of dense boxes that remain after Step~\ref{stp:sparsification} in~\cref{algo:main} contains a matching $M'$ of $\phi$-legged $\nu$-appearances of cardinality at least~$\frac{9\eps |S|}{10k}$.
\end{claim}
\begin{proof}
 Since $D_{m',m'}$ contains at most $\kappa m'$
 marked boxes, it follows that at most $\frac{\kappa}{d} = \frac{\eps}{100k}$ fraction of the layers have more than $d$ marked
 boxes. Hence, using the bound on the density of each layer (due to the
 nice $m$-partition of $\bx(S,I)$) from~\cref{cl:q_layering}, deleting the points in these marked boxes deletes at most $\frac{\eps}{100k
   } m' \cdot  \frac{4}{m} \cdot |S| \leq \frac{8 \eps |S|}{100k
   }$ points from $D$.  By a similar argument, the number of points that get deleted 
   by removing stripes with more than $d$ marked boxes is at most $\frac{\eps |S|}{100k
   }$. 
 Moreover, by the third item of~\cref{cl:q_layering}, we know that
 the total number of points that belong to marked boxes that are not tagged dense by
 $\mathsf{AlgTest}$ is at most $\beta \frac{|S|}{m'} \cdot \kappa m' \leq \frac{\eps |S|}{200k}$, where the inequality follows by our setting of $\beta$. 
 Finally, combining the second item in~\cref{cl:q_layering} 
 with the fact that we delete each stripe containing more than
 $d = \frac{100k\kappa}{\eps}$ marked boxes, 
 for each stripe that is left, the marked boxes contain at least $1-1/(\log^2 n)$ fraction of the
 points in it. 
 Hence, the total number of
 points deleted in~Step~\ref{stp:sparsification}
of \cref{algo:main} is at most $ \frac{8 \eps |S|}{100k
   }  + \frac{\eps |S|}{100k
   }  + \frac{\eps |S|}{200k
   } + \frac{|S|}{\log^2 n} \leq \frac{\eps |S|}{10k}$.

   Recall that we assume  that $D$ is $\eps$-far from being
   $\phi$-legged $\nu$-free. This implies that it contains a matching of
 $\phi$-legged $\nu$-appearances of size at least $\eps |S|/k$. For the rest of this
 proof, we fix such a matching $M$. 
 Since each
 deleted point deletes at most $1$ member from $M$, there is a matching $M'$ of cardinality at least $\frac{\eps |S|}{k} -
 \frac{\eps |S|}{10k} \geq \frac{9\eps |S|}{10k}$ with all legs in the set of dense boxes remaining after~Step~\ref{stp:sparsification}.  
\end{proof}

 We can partition $M'$ into a collection of disjoint matchings $M' =  \bigcup_{i \in [r]} M_i$, where $M_i$ contains the $\phi$-legged $\nu$-appearances in $M'$ belonging to configuration copies in $D_{m',m'}$
 that have~$i$ components.  
 Recall that all the legs of every $\phi$-legged $\nu$-appearance in $M'$ belong to the single component $D$ made of the boxes $B_1, \dots , B_t$. 
 However, with respect to the grid $D_{m',m'}$, each such $\nu$-appearance has a corresponding leg mapping that maps the legs of the appearance to boxes in $D_{m',m'}$, which are actually subboxes of $B_1, \dots, B_t$.
 Some of the leg mappings of $\nu$-appearances to subboxes might result in configurations with multiple components in the finer grid $D_{m',m'}$.
 
 It follows that either $M_1$ or $\bigcup_{i \in [r-1]} M_{i+1}$ has cardinality
 at least $\frac{9\eps |S|}{20k}$.
 Let $\eps_1 = 9\eps/(20k)$.
 
 \begin{claim}
 If $|M_1| \geq \eps_1 |S|$, then with high probability,~\cref{algo:main} finds a $\phi$-legged $\nu$-appearance or a $\pi$-appearance in~Step~\ref{stp:single-comp}.
 \end{claim}
 \begin{proof}
 The number of
 $1$-component configuration copies in the grid $D_{m',m'}$ that share a dense box and contain at most $r$ boxes is at most $(r-1)! \cdot \left(2d\right)^{r-1}$.
 Combined with the fact that the total number of dense boxes is at most $dm'$, we can see that the number of distinct copies of $1$-component configurations with at most $r$ boxes
 is at most $dm' \cdot (r-1)! \cdot \left(2d\right)^{r-1}$. 
 
 Therefore, in expectation, a uniformly random copy of a $1$-component configuration with at most $r$ boxes contains at least $\frac{\eps_1|S|}{dm'\cdot (r-1)!(2d)^{r-1}}$ many $\nu$-appearances from $M_1$. 
 These $\nu$-appearances each could have different leg mappings that are each $(\phi,D)$-consistent (see~\cref{def:phi-consistence}).
 There are at most $r^r$ ways to map the $r$ legs of $\nu$ into at most $r$ boxes.
 Thus, in expectation, a uniformly random $1$-component copy $C$ and a uniformly random $(\phi,D)$-consistent mapping of $r$ legs into the boxes of $C$ correspond to at least $\frac{\eps_1|S|}{dm' \cdot (r-1)!(2d)^{r-1} \cdot r^{r}}$ many $\nu$-appearances from $M_1$.

 By the reverse Markov's inequality\footnote{Let $X$ be a random
   variable such that $\Pr[X \leq a] = 1$ for some constant $a$. Then,
   for $d < E[X]$, we have $\Pr[X > d] \geq \frac{\E[X] -d}{a - d}$.},
 with probability at least
 $\frac{\eps_1}{(2d)^{r}\cdot (r-1)! r^{r}}$, the number of
 $\phi'$-legged $\nu$-appearances in a uniformly random one-component
 configuration copy for a uniformly random $(\phi,D)$-consistent leg
 mapping $\phi'$ is at least
 $\frac{\eps_1|S|}{m'(2d)^{r}\cdot (r-1)! r^{r}}$.  
 \ignore{In the
   reverse Markov application, we have used $a$ to be equal to
   $|S|/m'$, which is the upper bound on the total number of
   appearances in any connected
   component.}
 Therefore, w.h.p., at least one of the $\frac{\log^3 n}{\eps^{r+1}}$
 sampled one-component configuration $C$ and an associated
 $(\phi,D)$-consistent leg-mapping $\phi'$ contains at least
 $\frac{\eps_1|S|}{m'(2d)^{r}\cdot (r-1)! r^{r}}$ many $\phi'$-legged
 $\nu$-appearances. Conditioned on this event, the sub-grid restricted to
 this component is at least
 $\frac{\eps_1}{(2d)^{r}\cdot (r-1)!r^{r}}$-far from being
 free of $\phi'$-legged $\nu$ appearances. By
 the induction hypothesis, the recursive call in~\cref{algo:main} in~Step~\ref{stp:single-comp} with parameter $\eps'' = \frac{\eps_1}{(2d)^{r}\cdot (r-1)!r^{r}}$ will detect, with high probability, one such $\phi'$-legged
 $\nu$-appearance, which is also a $\phi$-legged $\nu$-appearance in $D$.
\end{proof}

Next, we consider the case that $|\bigcup_{i \in [r-1]} M_{i+1}| \geq \eps_1 \cdot |S|$. We start by presenting the following definitions and claims.
 Let $\ell$ be a positive integer. Let $H \subseteq
   [\ell]^t$ be a collection of weighted ordered $t$-tuples.  
   Each $t$-tuple $\mathbf{a} \in H$ is associated with a positive weight $w(\mathbf{a})$.
   We use $w_H$ to denote $\sum_{\mathbf{a} \in H} w(\mathbf{a})$, i.e., the sum of weights of all elements in $H$.
   Let $\mathbf{a}_i$ denote the $i$-th coordinate in $\mathbf{a}$.
   For $x \in [\ell]$ and $i \in [t]$, let $w_i(x) = \sum\limits_{\mathbf{a} \in H: \mathbf{a}_i = x} w(\mathbf{a})$. 
   Namely, $w_i(x)$ is the sum of weights of elements in $H$ that have the
   value $x$ in their $i$-th coordinate. 
   \begin{definition}\label{def:heavy}
   For $H$ with $w_H = \ell p$, we say that $x\in
   [\ell]$ is $(i,\alpha)$-heavy if $w_i(x) \geq \alpha p$. 
   \end{definition}
 \begin{claim}
   \label{clm:hypergraph}
   Let $H \subseteq [\ell]^t$ be such that $w_H = p
   \ell$. 
   Then for every $\alpha \leq \frac{1}{t}$, there exists $\mathbf{a} \in H$
   such that for every $i \in [t]$, $\mathbf{a}_i$ is $(i,\alpha)$-heavy.
 \end{claim}
 \begin{proof}
 For any fixed $i \in [t]$, $\sum_{x \in [\ell]} w_i(x) = w_H = p \ell$. On the
 other hand, the sum of weight of values $x \in [\ell],$ each that is not $(i, \alpha)$-heavy
 is less than $\ell \cdot \alpha p$ by definition.
 Therefore, the set of all $\mathbf{a}
 \in H$ in which for some $i \in [t]$, $\mathbf{a}_i$ is not $(i,\alpha)$-heavy
 has a total weight  less than $t\ell\cdot \alpha p 
 \leq p\ell$, since  $\alpha \leq \frac{1}{t}$.
 Hence, removing all such tuples from $H$ leaves at least one tuple $\mathbf{a}'
 \in H$. By definition,  for every $i \in [t],$ $\mathbf{a}'$ is $(i,\alpha)$-heavy. 
 \end{proof}

  \begin{claim}
  If $|\bigcup_{i \in [r-1]} M_{i+1}| \geq \eps_1 \cdot |S|$, then with high probability,~\cref{algo:main} finds a $\phi$-legged $\nu$-appearance or a $\pi$-appearance in~Step~\ref{stp:multi-comp}.
  \end{claim}
  \begin{proof}
  Consider $2 \leq h \leq r$ such that $|M_h| \geq \eps_1 |S|/r$.
  There are at most $c = r^{3r}$ configurations with $\leq r$ boxes along with their associated mappings of $r$ legs into those boxes. 
  Consider a $(\phi,\nu,D)$-consistent (see~\cref{def:phi-consistence}) $h$-component configuration $\mathcal C$ 
  along with corresponding leg mappings $\{\phi_i\}_{i \in [h]}$ and sub-patterns $\{\nu_i\}_{i \in [h]}$ such that there are at least $|M_h|/c$ appearances in $M_h$
  that form the configuration $\mathcal{C}$.
  Let $V$ be the set of
 all $1$-component configuration copies made up of at most $r$ dense boxes in $D_{m',m'}$. 
 Let $H\subseteq V^h$ be the (hyper)graph
  over the vertex set $V$, where $(C_1,C_2 \dots C_h) \in H$ if it is a copy of $\cC$. 
 Its weight, denoted $w(C_1,C_2 \dots C_h)$, is the number of  
 $\phi$-legged $\nu$-appearances in $M_h$ such that each one of them decomposes, for $i \in [h]$, into $\phi_i$-legged $\nu_i$-appearances in $C_i$.
 It must be the case that $w_H \geq \eps_1 |S|/(cr)$, as each member in~$M_h$ forming the configuration $\mathcal C$ contributes to $w_H$.
 
 This corresponds to the
 setting of~\cref{clm:hypergraph} with $t = h$ and $\ell = |V|$. In
 addition, since $\ell = |V| \leq dm' \cdot (2d)^{r-1} \cdot (r-1)!,$ it follows
 that
 $w_H = \ell \cdot \frac{w_H}{\ell} \geq \ell \cdot \frac{\eps_1}{c 2^{r-1} d^r r!} \cdot
 \frac{|S|}{m}$ which corresponds to $p = \frac{\eps_1}{c 2^{r-1} d^r r!} \cdot
 \frac{|S|}{m}$.  Finally,
 for $\tilde{C} \in V$ and $i \in [h]$, the quantity $w_i(\tilde{C})$ is the number of
 $\phi$-legged $\nu$-appearances in $M_h$ forming the configuration $\mathcal{C}$, where the legs of the $\nu_i$ subpattern are $\phi_i$-mapped to
 $\tilde{C}$. 
 
 Let $\alpha = \frac{1}{r} \leq \frac{1}{h} = \frac{1}{t}$ for the application of~\cref{clm:hypergraph}.
 As a result,~\cref{clm:hypergraph} guarantees that there is
 an $h$-tuple $(C_1,C_2, \dots C_h) \in V^h$ consistent
 with a copy of $\mathcal C$, contains a $\phi$-legged $\nu$-appearance consistent with $\mathcal{C}$, and for which, each of $C_1, C_2 \dots C_h$ are
 $\alpha$-heavy (see~\cref{def:heavy}).  
 In turn, this means that for each $i \in [h]$, the component $C_i$
 is 
 $\eps'$-far from being $\phi_i$-legged $\nu_i$-free, where $\eps'
 = \frac{p\alpha}{r|S|/m} = \frac{\eps_1}{cr^2 \cdot r! \cdot 2^{r-1} \cdot d^r}$.
 This additionally implies that the overall density of the marked boxes
 involved in each $C_i, i \in [h]$ is also at least $\eps'$.    
 This, in turn, implies that each test for
 $\phi_i$-legged $\nu_i$-freeness in $C_i$ for $i \in [h]$ (with the corresponding
 distance parameter $\eps'$), that is done in Step~\ref{stp:multi-comp}, is going to
 succeed with very high probability by the induction hypothesis. 
\end{proof}

This concludes the proof of correctness for all cases, assuming
inductively correctness for $r-1$ and $a-1$. Since at each recursive
call we decrease $a$ by $1$, the recursion
depth is at most $a$. In particular, as the distance
parameter $\eps$ decreases by a constant factor (assuming $k$ and $a$ are constants), namely, that is
independent of $n$ or $|S|$, we conclude that it remains constant at
all recursion levels.

\paragraph{Query Complexity.} We now analyze the query complexity of a call to $\mathsf{AlgTest}_\pi(\nu,\phi,D,m,\eps)$. 
We fix $m = m(|S|)$ and write the
 complexity in terms of $m$, where $S$ is the set of indices of $D$ as defined in the algorithm.
 For this fixed $m$ (that is not going to be
 changed during recursive calls), let 
 $a$ be the
 smallest integer for which $m^{a} \geq k^{a-1} |S|$.  

\begin{lemma}
\cref{algo:main} makes $\tilde{O}\left(a! \cdot m^r \cdot \left(\frac{k}{\eps}\right)^{\Theta(k^a)}\right)$ queries.
\end{lemma}
\begin{proof}
  
 We express the
 query complexity as $q(m,a,r,\eps)$, where the integer $a$ here denotes the recursion depth.  
Recall that $S, I$ are as defined in the algorithm (based on the component~$D$ contains).
One of the boundary cases is when $\nu$ has length $2$ and we have $q(m,a',2,\eps) = O\left(\frac{\log |S|}{\eps}\right) = \tilde{O}(1/\eps)$ from~\cref{lm:2-pattern} for $a' \leq a, \eps \in (0,1)$.
 The other boundary case is for $q(m,1,r',\eps)$ for $\eps \in (0,1), r' \leq r$, i.e., the query complexity after $a-1$ recursive calls. Since the size of the domain on which the algorithm is called decreases by a factor of at most $k/m$ with each recursive call, the final domain size is at most $\left(\frac{k}{m}\right)^{a - 1} \cdot |S| \leq m$ and the algorithm can query every index in the domain at that level. Therefore $q(m,1,r',\eps) = m$ for all $\eps \in (0,1), r' \leq r$. 
 

We now write the recurrence for the general case by ignoring polylog factors.
 We have $$q(m,a,r,\eps) = m\cdot \frac{k^2}{\eps^2} + m\cdot \left( \frac{k}{\eps}\right)^{\Theta(k)} q(m,a-1,r-1, \eps^{(1)}) +  \left( \frac{1}{\eps}\right)^{\Theta(k)} q(m,a-1,r,\eps^{(1)}),$$ where $\eps^{(1)} = \left( \frac{\eps}{k}\right)^{\Theta(k)}$. The
 first term on the right comes from the gridding of $D$. The second term comes from
 recursively calling the algorithm for each one of the possible $m$ components in
 $D$ for the constantly many smaller patterns and constantly many consistent leg-mappings (Step~\ref{stp:multi-comp}), and the third item from the self
 recursion in Step~\ref{stp:single-comp}.  
\ignore{ 
 \nvedit{The first term in the recurrence is dominated by the second term and hence we get the following simplified recurrence.
 \begin{align*}
 q(m,a,r,\eps) = m\cdot \left( \frac{k}{\eps}\right)^{\Theta(k)} q(m,a-1,r-1, \eps^{(1)}) +  \left( \frac{1}{\eps}\right)^{\Theta(k)} q(m,a-1,r,\eps^{(1)})
 \end{align*}

 For integer $1 \leq a' \leq a$, let $\eps^{(a')}$ denote the value $\left( \frac{\eps^{(a'-1)}}{k}\right)^{\Theta(k)} = \left( \frac{\eps}{k}\right)^{\Theta(k^{a'})}$, where $\eps^{(0)}$ is defined to be $\eps$.
 For substitution, we assume that $q(m,a,r,\eps) \leq a! m^r \left(\frac{k}{\eps}\right)^{\Theta(k^a)}$. 
 
 \begin{align*}
 q(m,a,r,\eps) = m\cdot \left( \frac{k}{\eps}\right)^{\Theta(k)} q(m,a-1,r-1, \eps^{(1)}) +  \left( \frac{1}{\eps}\right)^{\Theta(k)} q(m,a-1,r,\eps^{(1)})
 \end{align*}
 
 In the above, we first keep expanding only the second terms repeatedly, until all terms have $r - 1$ in their arguments. There will be at most $a$ terms and we will take a really worst case sum and get $am \left(\frac{k}{\eps}\right)^{\Theta(k^{a})} \cdot q(m,a-1,r-1,\eps^{(a-1)})$. By substituting our assumed upper bound, we will get what we want to prove.
 }}

 We can solve the recurrence above to get the solution $q(m,a,r,\eps) = a! \cdot m^r \cdot \left(\frac{k}{\eps}\right)^{\Theta(k^a)}$.
 As long as $a \leq \log \log \log n$, the distance parameter at the lowest recursion levels is $\Omega(1/\log n)$, which is allows us to call
 the $2$-pattern testers (one of the base cases) in $\tilde{O}(1)$ query complexity.  
  \end{proof}




 \paragraph{A note on single-valued layers:}  
 In case the function is not one-to-one, the subroutine Layering may discover single-valued layers if
 there are such values of large enough density. The first place where
 this might have an impact is at Step~\ref{stp:sparsification} of $\mathsf{AlgTest}$, where we delete
 all layers that contain more than a constant number of marked boxes. The modification is that we
 apply this step only on multi-valued layers. Hence, single-valued
 layers will possibly remain with more than $d$ marked
 boxes. However,~\cref{cl:q_layering} remains unchanged. 

 We used the fact that every layer has at most $d$ marked boxes
 to conclude that every marked box that may be involved in a
 $\nu$-appearance is contained in at most $(r-1)!(2d)^r$
 components of at most $r$ boxes. This claim is still true since a
 $\nu$-appearance cannot have two or more legs in a single
 valued layer. Hence,
 each connected component that contains a $\nu$-appearance, can contain at most one box from each single-valued layer.
  It follows that  all our arguments go through even in this setting. The only
 difference is that we also consider components containing at most one marked box
 per single-valued layer.

\paragraph{Acknowledgements.} We thank anonymous reviewers for their extensive comments that significantly improved the presentation of the manuscript.


\newpage

\printbibliography
\end{document}